\documentclass[12pt,draftclsnofoot,onecolumn]{ctextemp_IEEEtran}
\usepackage{amssymb}
\usepackage{amsmath}
\usepackage{amsmath,bm}
\usepackage{amsthm}
\usepackage{algorithm}
\usepackage{graphicx}
\usepackage{subfigure}
\usepackage{cite}
\usepackage{enumerate}
\usepackage{url}
\usepackage{bm}
\usepackage{algorithmicx,algorithm}
\usepackage{setspace}
\usepackage{color, soul}
\usepackage{epstopdf}
\usepackage[noend]{algpseudocode}
\usepackage{multirow}

\newtheorem{Def}{Definition}

\newtheorem{prop}{Proposition}
\newtheorem{remark}{Remark}
\makeatother

\begin{document}

\iffalse
\newtheorem{definition}{\bf~~Definition}
\newtheorem{theorem}{\bf~~Theorem}
\newtheorem{observation}{\bf~~Observation}
\newtheorem{proposition}{\bf~~Proposition}
\newtheorem{remark}{\bf~~Remark}
\newtheorem{lemma}{\bf~~Lemma}
\fi

\title{\Huge{%HetMEC: Task Assignment and Resource Allocation in Heterogeneous Mobile Edge Computing Networks for Latency Minimization
HetMEC: Latency-optimal Task Assignment and Resource Allocation for Heterogeneous Mobile Edge Computing
}}
\author{
\IEEEauthorblockN{
{Pengfei Wang}, \IEEEmembership{Student Member, IEEE},
{Zijie Zheng}, \IEEEmembership{Student Member, IEEE},\\
{Boya Di}, \IEEEmembership{Student Member, IEEE},
%{Xiaoyang Wang},
%{Guangyu Sun}, \IEEEmembership{Member},
{and Lingyang Song}, \IEEEmembership{Fellow, IEEE}}\\

%\vspace{-0.5cm}

\thanks{The authors are with School of Electronics Engineering and Computer Science, Peking University, Beijing, China (email: wangpengfei13@pku.edu.cn; zijie.zheng@pku.edu.cn; diboya@pku.edu.cn; lingyang.song@pku.edu.cn).}
}

\maketitle

\vspace{-0.5cm}

\begin{abstract}
Driven by great demands on low-latency services of the edge devices (EDs), mobile edge computing~(MEC) has been proposed to enable the computing capacities at the edge of the radio access network. However, conventional MEC servers suffer disadvantages such as limited computing capacity, preventing the computation-intensive tasks to be processed in time. To relief this issue, we propose the heterogeneous MEC (HetMEC) where the data that cannot be timely processed at the edge are allowed be offloaded to the upper-layer MEC servers, and finally to the cloud center (CC) with more powerful computing capacity. We design the latency minimization algorithm by jointly coordinating the task assignment, computing and transmission resources among the EDs, multi-layer MEC servers, and the CC. Simulation results indicate that our proposed algorithm can achieve a lower latency and higher processing rate than the conventional MEC scheme.

\end{abstract}


\newpage

%%%%%%%%%%%%%%%%%%%%%%%
\section{Introduction}%
%%%%%%%%%%%%%%%%%%%%%%%
With the rise of the Internet of Things (IoT), namely a network including interconnected devices capable of exchanging information~\cite{LAG-IOT-2010,DSFI-2012}, huge amount of data are generated and transmitted throughout the communication networks~\cite{CISIO}.
%It is predicted by CISCO that 50 billion IoT devices will connect to the Internet by 2020.
%With the increasing number of electronic devices appear in our lives, the Internet of Things (IoT) rises in the recent years, namely a network of interconnected devices with the ability of exchanging information~\cite{LAG-IOT-2010,DSFI-2012}.
%With the development of IoT and mobile internet, huge volume and varieties of data is generated and transmitted into the network.
However, the computing capacities of the current communication networks are not sufficient to satisfy users' increasing demands on high data rates~\cite{ABE-2015}.
Traditionally, cloud computing has been proposed as an effective solution for such data explosion by making use of the strong computing capacity of the data center~\cite{CLOUD}. As a centralized paradigm, cloud computing can provide a wide range of services and massive computing resources supported by a large group of computers in the data center. However, the data transmission from the edge of the network to the remote cloud center usually induces high latency, which is unacceptable for the latency-sensitive applications~\cite{ETSI,Vehicle}.
%Moreover, the bandwidth and transmission resources of the wired and wireless networks are limited. When huge amount of data needs to be processed, the limited bandwidth cannot afford such amount of data transmission and becomes the bottleneck of the cloud computing.
%Therefore, the additional load to the transmission links and long latency caused by the cloud computing cannot fulfill the demands of multiple emerging mobile applications.
\par
To deal with the dilemma of cloud computing, mobile edge computing (MEC) has been investigated, which enables the computation to be performed at the mobile devices and the access points (APs)\footnote{Base stations are the typical APs in the radio access networks.} within the radio access networks~\cite{Kaibin-MECsurvey,MECsurvey-2011}.
The MEC servers that possess the computing resources, e.g., the APs, offer rich services in close proximity to the end users, also known as the edge devices (EDs)~\cite{CLOUDLET}.
When these EDs generate computation tasks at the edge of the communication networks, they can offload tasks to the MEC servers nearby rather than the remote cloud center~\cite{MT-2016}.
Therefore, the MEC provides the low-latency and high-efficient data processing due to the proximity of the computing resources~\cite{WJQYL-2016,SXYLJW-2017}.
However, most works only consider the MEC servers that directly communicate with the EDs via wireless links to offer the in-proximity services\cite{LJHY-2015,XuChen-2016,ShuguangCui-2018}.
Due to the limited computing capacities of these MEC servers, it would be desirable that the data that cannot be processed at the MEC servers can be further offloaded to the upper-layer MEC servers, until to the cloud center (CC).
%It cannot be neglected that the computing capacities of these MEC servers are still limited.

% introduce the heterogeneous mobile edge computing network
In this paper, we consider \emph{heterogeneous MEC} (HetMEC) for uplink communications, where EDs divide and offload the computational intensive tasks to multi-layer MEC servers and the CC for latency performance improvement~\cite{ETSI}.
Classified by the locations in the wired-wireless networks, various function nodes with certain computing capacities serve as the \emph{MEC servers} on different layers, that is, the APs, switches, network gateways, and small data centers from the bottom up~\cite{3GPP-CoreNetwork}.
%Specifically, the multiple layers of \emph{MEC servers} include various function nodes with certain computing capacities in the wired-wireless networks, which are specifically the APs, switches, network gateways, and small data centers.
The data flow of each ED first transits in the radio access networks via a wireless ED-AP link. The received data of each AP are then partially processed and delivered to the wired core network,  passing through the switches, network gateways, and the small data centers sequentially~\cite{ComputerNetwork}. Locating at bottom-up layers, these MEC servers provide increasing computing capacity for data processing and finally send the data from the bottom layers to the remote CC\footnote{It is worth noting that the switches, network gateways or small data centers are not all necessary in the wired networks. The APs are possible to connect with the network gateways directly, and the network gateways may also connect with the remote cloud center.}.
Based on such a HetMEC structure, the computing resources of multi-layer MEC servers and the CC can be fully exploited to support computation-intensive and latency-sensitive tasks of the EDs with strong robustness.

\iffalse
where the EDs connect with the CC through multiple layers of MEC servers \cite{3GPP-CoreNetwork,ETSI}.
Specifically, various function nodes with certain processing capacities in the wired-wireless networks can serve as the MEC servers, such as the BSs, switches, network gateways, and small data centers. %Each ED can then be connected with the remote cloud center (CC) through multiple layers of MEC servers, forming the heterogeneous MEC (HetMEC) network.
%The MEC servers locate at different places in the wired-wireless networks.
The raw data generated at the EDs are partially divided and offloaded to multiple layers of MEC servers and the CC, and the results are aggregated at the CC.
The EDs accesses the network by the BSs via the wireless links between them. After collecting the data from the EDs in the radio access networks, the BSs deliver the data to the switches, which store and forward the received data to other switches or the network gateway within the local area networks (LANs)~\cite{ComputerNetwork}. The network gateways transmit the data of the LANs to the small data centers, which provide relatively strong computing capacity for multiple LANs and connect with the remote CC\footnote{It is worth noting that the switches, network gateways or small data centers are not all necessary in the wired networks. The BSs are possible to connect with the network gateways directly, and the network gateways may also connect with the remote cloud center.}.
\fi

\par A number of challenges induced by the heterogeneous nature of the multi-layer MEC networks still remain to be solved. %Though the utilization of computing resources in the current MEC network architecture has been initially considered,
\emph{First}, since the data of each task can be divided and partially processed by multiple MEC servers on different layers, the task assignments among these servers are coupled with each other by the limited resources of their own. In other words, the amount of offloaded data in one MEC layer is correlated with that in all the other layers, which is different from that in the traditional MEC networks\footnote{In the traditional MEC networks, EDs can offload data to only one layer of MEC servers, i.e., the APs.}.
%Unlike that in the traditional two-layer networks, the task assignment of each MEC layer is correlated with all the other layers, rather than only adjacent layers, so that the computing resources of the whole HetMEC network can be fully utilized.
\emph{Second}, the transmission resource allocation in both the wireless and wired network need to be considered, which are closely related with the task assignment among multiple layers of MEC servers. To be specific, the allocated wired transmission resources directly restrict the data transmission rate between adjacent layers of MEC servers.
\emph{Third}, due to the limited computing capacity of each MEC server, the robustness of the HetMEC network should be considered and evaluated in response to the time-varying data generation speed at the EDs.

\par
In the literature, the above challenges induced by the HetMEC architecture have not been fully addressed~\cite{Hierarchical-2016,AN-2017,BZGTHQ-2015}. Most existing works either do not consider the task assignment, computing and transmission resource allocation jointly~\cite{CMRG-2015,Kaibin-2017,Tang-2017,Kaibin-2018,XHCP-2017}, or fail to depict the relations between multiple layers in the HetMEC network~\cite{TD-2018,CFCQL-2017}.
%A few works discuss the task assignment~\cite{LJHY-2015,CMRG-2015}, some together with transmission resource allocation~\cite{XuChen-2016,Kaibin-2017,Tang-2017} or computing resource allocation~\cite{Kaibin-2018,XHCP-2017}.
%In \cite{TaskOffloading-2017}, a VOP task assignment strategy is designed taking the execution time, energy consumption and other expenses into consideration.
In~\cite{CMRG-2015}, an efficient $k$-out-of-$n$ task assignment scheme is proposed to minimize the execution time on multiple processor nodes and save energy consumption.
In~\cite{Kaibin-2017}, the transmission resource allocation is studied for multi-user mobile edge computational offloading constrained by the computation latency. %Authors in \cite{Tang-2017} discuss the computation offloading, resource allocation and content caching strategy in wireless cellular networks.
Authors in~\cite{Kaibin-2018} analyze the transmission latency and computation latency separately, taking the task assignment and computing rate control into account.
In~\cite{TD-2018}, the traditional MEC networks are discussed, where the MEC computing resource allocation and uplink power allocation are studied along with the binary\footnote{The binary task offloading means that the task is impartible and is processed at either the edge device or the MEC server.} task offloading.
In~\cite{CFCQL-2017}, the computation offloading and interference management are performed in the wireless cellular networks with a single MEC server.
Unfortunately, joint task assignment among multi-layer MEC servers in the HetMEC network has not been taken into account together with the computing and transmission resource allocation.

%\par Unlike the above works, we consider a HetMEC network with a multi-layer tree structure.
%We aim to minimize the latency of processing all tasks through jointly optimizing the task assignment, computing and transmission resource allocation among multiple layers. A global-optimal latency minimization algorithm (LMA) is then proposed to achieve this goal in the HetMEC network.

The main contributions of our paper are summarized as follows.
\begin{itemize}
\item %We study an uplink HetMEC network where EDs divide and offload the computational intensive tasks to multiple layers of MEC servers and the CC for latency performance improvement. The data generated at the EDs are finally aggregated at the CC.
    We study the HetMEC network consisting of the EDs, multiple layer of MEC servers and the CC. The uplink transmission is considered where the data generated at the EDs are processed at multiple layers of MEC servers and finally aggregated at the CC.
    %The non-congested constraints of the HetMEC network are derived.
\item In order to improve the latency performance of the HetMEC network, we jointly considering the task assignment, computing and transmission resource allocation among multiple layers. A latency minimization algorithm (LMA) is then developed to achieve the global-optimal system latency.
\item Simulations are performed in the considered HetMEC networks with different numbers of layers, and the results show that our algorithm LMA achieves a lower latency and higher processing rate than the previous schemes. The influence of the number of MEC layers on the robustness performance has been discussed.
\end{itemize}

\par
The rest of the paper is organized as follows.
In Section \ref{sec:model} we describe the system model of the HetMEC network.
In Section \ref{sec:formulation} we discuss the system constraints of the HetMEC network and formulate the system latency minimization problem.
To solve this problem, we design the algorithm LMA and analyze the influence of the number of MEC layers on the network robustness in Section \ref{sec:algorithm_design}.
Simulation results are given in Section \ref{sec:simulation_results}. Finally the conclusions are drawn in Section \ref{sec:conclusion}.

%%%%%%%%%%%%%%%%%%%%%%%%%%%%%%%%%%%%%%%%%%%%%%%
\section{System Model \label{sec:model}}%
%%%%%%%%%%%%%%%%%%%%%%%%%%%%%%%%%%%%%%%%%%%%%%%
\begin{figure}[t]
\centering
\includegraphics[width=0.8\textwidth]{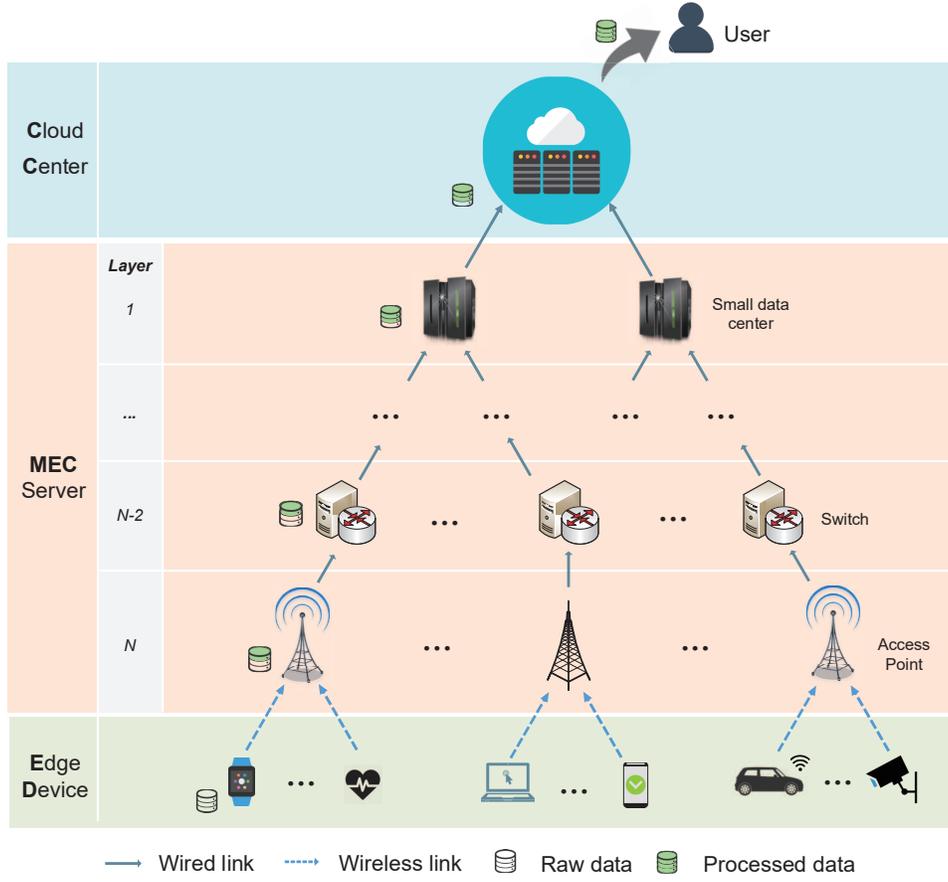}
\caption{The architecture of the HetMEC network.}
\label{Fig:model}
\end{figure}

As shown in Fig.~\ref{Fig:model}, we consider a HetMEC network consisting of the EDs, $N$ layers of MEC servers, and a CC.
Each ED accesses the network by the AP through the wireless links between them. The uploaded data received by the EDs can then be forwarded to the MEC servers on the upper layers, connecting to the CC via wired links.
%For convenience, we define the CC as layer $0$, and the EDs as layer $N+1$.
The number of devices on each layer $n$ is denoted by $M_n$, $1\leq\! n \!\leq N$.
In such a tree structure, each node (ED or MEC server) connects with at most one parent node in the upper layer. The number of the child nodes\footnote{That is to say, $\sum_{i=1}^{M_{n-1}}Q_i^{n-1}=M_n$.} connected with the parent node $i$ on layer $n$ is denoted by $Q_n^i$, and the set of the child nodes is denoted by $\mathcal{Q}_n^i$.
Each node has a different computing capacity. To communicate with its child nodes, each MEC server and the CC possesses a certain amount of wireless or wired transmission resources\footnote{Among the MEC servers, the APs possess the wireless transmission resources, while the others possess the wired transmission resources.}.
We assume that all nodes access the wireless or wired channel via time division multiple address~(TDMA) technology, implying that the frequency bands occupied by any two APs are orthogonal.

\par For a typical uplink MEC application, where the raw data are generated at the EDs and the results of the data processing need to be aggregated at the CC.
%The process of the data can be performed at any layer from the EDs to the CC\footnote{The EDs, APs, MEC servers and the CC all have the computing resources, and are able to process the raw data.}, and the percentage of the data to process at each node is adjustable.
The task generated at the ED is divided into multiple parts, and the ED, each MEC server on different layers or the CC only processes a part of it\footnote{The EDs, MEC servers and the CC all can process the raw data.}. After processing its own part, the processing results and the remaining raw data are delivered to the upper layer. The percentage of the data to be processed at each node is adjustable.
Moreover, once the data are processed at the edge of the network, i.e., at the ED or MEC servers, the output results, which are then forwarded to the CC, usually have a much smaller size than the raw data.

The mathematical models of data processing and transmitting at the ED, MEC server and CC are listed below.

\subsection{Edge Device}
The EDs, including the cars, smartwatches, cameras, etc., are on the bottom of the HetMEC networks, and usually responsible for generating the raw data.
%The EDs on the bottom layer, e.g., the laptops, cars, smartwatches, and cameras, are responsible for generating the raw data.
For convenience, we refer to the EDs as layer $N+1$.
Each ED processes part of the raw data, and delivers the results of the raw data together with the rest raw data to the node (AP) on the $N$-th layer via wireless link. Let $s_{N+1}^{i}$ represent the task division percentage of ED $i$, which satisfies that
\begin{align}
0\leq s_{N+1}^{i}\leq 1.  \label{NBlockCon_ED1}
\end{align}
Let $\lambda_{N}^{i}$ denotes the data generation speed of ED $i$, and $\rho$ denotes the compression ratio after the data processing.
The computing capacity and the wireless transmitting capacity of ED $i$ connected with AP $j$ on the upper layer per unit time is denoted by $\theta_{N}^{i}$ and $\phi_{N}^{j,i}$, respectively.
The computing data volume is limited by its computing capacity.
\begin{align}
\lambda_{N+1}^{i} s_{N+1}^{i} \leq \theta_{N+1}^{i}, \label{NBlockCon_ED2}
\end{align}
and the maximum computing capacity that ED $i$ can offer is denoted by $\theta_{N+1}^{i,u}$, and thus,
\begin{align}
\theta_{N+1}^{i} \leq \theta_{N+1}^{i,u}.
\end{align}
The transmitting data volume is restricted by the wireless transmitting capacity of ED $i$, which is closely related with the wireless transmission resources allocated by AP $j$.
\begin{align}
\rho\lambda_{N+1}^{i} s_{N+1}^{i}+\lambda_{N+1}^{i}(1 - s_{N+1}^{i}) \leq \phi_{N+1}^{j,i}, \label{NBlockCon_ED3}
\end{align}
where $\rho\lambda_{N+1}^{i} s_{N+1}^{i}$ is the processing results, and $\lambda_{N+1}^{i}(1 - s_{N+1}^{i})$ represents the remaining raw data to transmit to AP $j$.
The total transmitting data volume of all EDs connected with AP $j$ is \emph{linearly} constrained by the wireless transmission resources of node $j$ on the $N$-th layer, denoted by $\phi_{N}^{j}$, which can be expressed by
\begin{align}
\sum_{i=1}^{M_j^{N}}\phi_{N+1}^{j,i} \leq \phi_{N}^{j}. \label{NBlockCon_ED4}
\end{align}
\begin{remark}
The constraints in~(\ref{NBlockCon_ED4}) can describe such wireless resources which influence the wireless data rate in a linear manner, e.g., the spectrum and time resources. The power and the antenna resources cannot be modeled in the similar way, which are left in the future works.
\end{remark}

\subsection{Mobile Edge Computing (MEC) Server}
The MEC servers share the computing pressure of the EDs.
Being the bottom layer of the MEC servers, the APs connect with EDs via wireless links and enable the EDs to access the wired networks.
Other MEC servers, e.g., the switches and network gateways, receive the raw data and the processing results from APs. After processing part of the receiving raw data, the MEC server forwards the rest raw data together with the processing results to its parent node (upper-layer MEC server or CC).

%model of wired network
When multiple MEC servers connect to the same upper-layer MEC server or belong to the same switch or bridge, the upper-layer MEC server can coordinate the connected MEC servers and allocate the transmission resources in a centralized way. Then transmission resources, e.g., the bandwidth or time, can be divided linearly to the multiple nodes \cite{ref:wiredmodel}.

As shown in Fig. \ref{Fig:model_dataflow}, we consider the node $j$ on the layer $n, 1\leq n\leq N$, which is connected with the node $k$ (MEC server or CC) on the layer $n-1$.

\begin{figure}[t]
\centering
\includegraphics[width=0.9\textwidth]{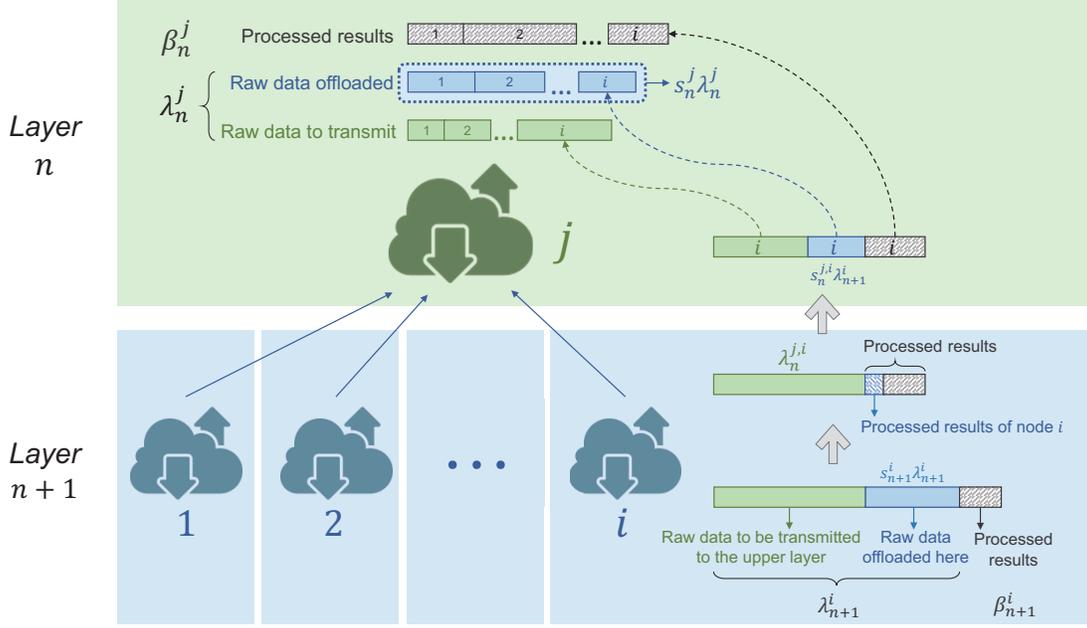}
\caption{The data processing and transmitting between two adjacent MEC layers of the HetMEC network.}
\label{Fig:model_dataflow}
\end{figure}

The raw data arrival speed from its child node $i$ on the $(n+1)$-th layer can be expressed by
\begin{equation}\label{Lambda_AP}
  \lambda_{n}^{j,i} = \phi_{n+1}^{j,i}\cdot\frac{(1-s_{n+1}^{i})\lambda_{n+1}^{i}}{(1-s_{n+1}^{i}+\rho s_{n+1}^{i})\lambda_{n+1}^{i} + \beta_{n+1}^{i}},
\end{equation}
where $\phi_{n+1}^{j,i}$ is the total data arrival speed to node $j$ from node $i$ on the lower layer, and only part of it is the raw data arrival speed. The raw data volume transmitted to node $j$ from its child node $i$ is $\lambda_{n+1}^{i}(1-s_{n+1}^{i})$, and the processed data volume is $\lambda_{n+1}^{i} \rho s_{n+1}^{i} + \beta_{n+1}^{i}$, where $s_{n+1}^{i}$ is the equivalent task division percentage\footnote{When the node $i$ is not at the bottom layer, i.e., it is not an ED, it may receives raw data from multiple links. $s_{n+1}^{i}$ is the percentage of the total raw data volume to be processed at node $i$ in its total received raw data volume.} of node $i$ and $\beta_{n+1}^{i}$ represents the volume of the processed data under the $(n+1)$-th layer.
Hence, the equivalent raw data arrival speed at node $j$ on the $n$-th layer is denoted by
\begin{equation}\label{Lambda_AP_equi}
  \lambda_{n}^{j} = \sum_{i\in \mathcal{Q}_{n}^{j}}\lambda_{n}^{j,i}.
\end{equation}
Accordingly, the total volume of the processed data received by node $j$ on the $n$-th layer can be expressed by
\begin{equation}\label{Beta_AP}
  \beta_{n}^{j} = \sum_{i=1}^{N_j}\left( \beta_{n+1}^{i} + \rho \cdot s_{n+1}^{i} \lambda_{n+1}^{i} \right).
\end{equation}

Let $s_{n}^{j,i}$ denotes the task division percentage of node $j$ for the data delivered from node $i$ on the $(n+1)$-th layer, which satisfies that
\begin{align}
0\leq s_{n}^{j,i}\leq 1, i\in\mathcal{Q}_{n}^{j}.  \label{NBlockCon_AP1}
\end{align}
The computing capacity of node $j$ is denoted by $\theta_{n}^{j}$.
The total data volume to be processed at node $j$, denoted by $C_n^j$, is limited by its computing capacity, which can be expressed by
\begin{align}
C_n^j = \sum_{i\in \mathcal{Q}_{n}^{j}} s_{n}^{j,i} \lambda_{n}^{j,i} \ \leq \ \theta_{n}^{j}.\label{NBlockCon_AP2}
\end{align}
Constrained by the limited computing resources, the maximum computing capacity that node $j$ on the $n$ layer can offer is denoted by $\theta_{n}^{j,u}$, and thus,
\begin{align}
\theta_{n}^{j} \leq \theta_{n}^{j,u}.
\end{align}
%where $\theta_{n}^{j,u}$ denotes the upper bound of the computing capacity of the node $j$ on the $n$ layer.
\iffalse
Therefore, the equivalent task division percentage of node $j$ on the $n$-th layer can be expressed by
\begin{equation}\label{S_AP_equi}
  s_{n}^{j} = \frac{C_n^j}{\lambda_{n}^{j}}.
\end{equation}
\fi

Let $\phi_{n}^{k,j}$ denote the wired transmitting capacity of node $j$ to node $k$. The transmitting data volume of AP $j$ is limited by its wired transmitting capacity, which is closely related with the wired transmission resources allocated by its parent node $k$ on the upper layer.
\begin{align}
\rho\lambda_{n}^{j} s_{n}^{j}+\lambda_{n}^{j}(1 - s_{n}^{j}) + \beta_{n}^{j}\leq \phi_{n}^{k,j}.\label{NBlockCon_AP3}
\end{align}
The data to be transmitted to node $k$ on the upper layer includes three parts. $\rho\lambda_{n}^{j} s_{n}^{j}$ is the processed data of node $j$, and $\lambda_{n}^{j}(1 - s_{n}^{j})$ is the remaining raw data to transmit, and $\beta_{n}^{j}$ is the processed data delivered from the lower layer. All the three parts need to be transmitted to the upper layer, which is limited by the allocated wired transmitting capacity $\phi_{n}^{k,j}$ of node $j$.
Moreover, the total transmitting data volume of all nodes on the $n$ layer connected with the node $k$ is limited by the wired transmission resources of the node $k$, denoted by $\phi_{n-1}^{k}$, which can be expressed by
\begin{align}
\sum_{j\in\mathcal{Q}_{n-1}^k}\phi_{n}^{k,j}\leq \phi_{n-1}^{k}. \label{NBlockCon_AP4}
\end{align}

\subsection{Cloud Center}
The CC collects the data from the MEC servers via wired links. All raw data delivered to the CC is processed and the whole results are forwarded to the user who generates the task. %Moreover, the CC determines the task assignment strategy of the whole network, that is, the task division percentage at each AP and ED.
For convenience, we refer to the CC as layer $0$.
\par The equivalent raw data arrival speed at the CC can be calculated by
\begin{equation}\label{Lambda_CC}
  \lambda_{0}^{1} = \sum_{i=1}^{M_0^{1}}\left[\phi_{1}^{1,i}\cdot\frac{(1-s_{1}^{i})\lambda_{1}^{i}}{(1-s_{1}^{i}+\rho s_{1}^{i})\lambda_{1}^{i} + \beta_{1}^{i}}\right],
\end{equation}
The arriving data at the CC includes three part: the remaining raw data, the processing results of the MEC servers and the processing results of the EDs. The raw data arrival speed is proportional to the remaining raw data volume percentage in the arriving data. Moreover, the computing capacity of the CC is denoted by~$\theta_{0}^{1}$, and the maximum computing capacity the CC can offer is denoted by $\theta_{0}^{1,u}$.

\par
\iffalse
In summary, the whole task starts from the data generation at the EDs and ends when all data is processed and transmitted to the CC.
\begin{itemize}
  \item After being generated at the EDs, part of the raw data are processed at the EDs, and the processing results together with the remaining raw data are transmitted to the APs.
  \item Once receiving the data from the corresponding nodes on the lower layer, the APs or MEC servers offload a part of the raw data to process, and deliver the left raw data, the processing results of themselves as well as the received processing results of the EDs to the upper layer (MEC server or CC) until the top layer, i.e., CC.
  \item The CC will process the remaining raw data and aggregate the processing results on different layers.
\end{itemize}
\fi
During the processing and transmitting from the EDs to the CC, the task assignment strategy $\bm{s}$, the computing capacity of each MEC server $j$ on the $n$-th layer $\theta_{n}^{j}$, the computing capacity of the CC $\theta_{0}^{1}$, the wireless transmission resources allocation $\phi_{N+1}^{j,i}$ and the wired transmission resources allocation of the $n$-th layer, $\phi_{n}^{k,j}$, need to be optimized, which will be discussed in Section \ref{sec:formulation}.

%%%%%%%%%%%%%%%%%%%%%%%%%%%%%
\section{Problem Formulation} \label{sec:formulation}%
%%%%%%%%%%%%%%%%%%%%%%%%%%%%%

In this section, we first analyze the system constraints of the considered HetMEC network, and then formulate the system latency minimization problem given these constraints.

%In this section, we first clarify that the system can be naturally classified into the congestion state and non-congested state, the boundary of which is also derived. Then, we formulate objectives for both states, respectively.

\subsection{System Constraints}
\begin{figure}[t]
\centering
\includegraphics[width=0.8\textwidth]{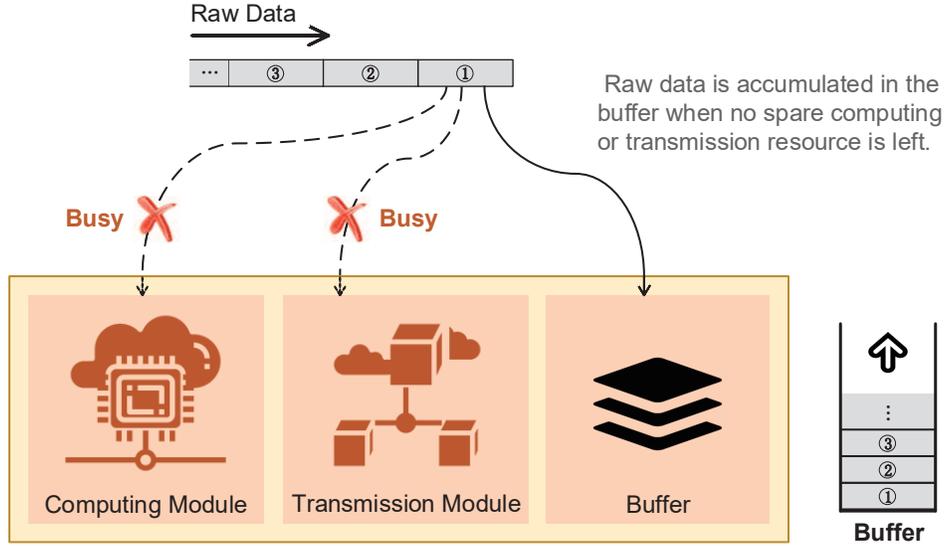}
\caption{An illustration of the congestion.}\label{fig:congestion_state}
\end{figure}
We first describe when and why a HetMEC network can be out of function due to the traffic congestion.
The total computing capacity of each node, the total wireless transmission resources of each AP, and the total wired transmission resources of each MEC server or the CC in our framework are finite, however, the data generation speed is fluctuant and time-varying.
As shown in Fig. \ref{fig:congestion_state}, when the data generation speed of the EDs exceeds a certain bound, the HetMEC network cannot follow up the data generation speed due to lack of available computing or transmission resources, and thus, the data will accumulate in the buffer\footnote{The space of the buffer in each node is viewed as infinite, that is, the data only accumulates in the buffer and no data loss happens when facing the congestion.}.
As the raw data keep accumulating, the waiting time in the buffer increase, eventually leading to a congested network.
%We explain the \emph{congestion} as follows.
\iffalse
\begin{remark}
The \textbf{congestion} is the state that no matter how the network adjusts its computing and transmission resources allocation or adjusts the task division on every node, the data will accumulates in the buffer of at least one node.
\end{remark}
\fi
\par
We then derive the system constraints of the HetMEC network under which the above congestion does not happen.
Specifically, each layer of the HetMEC network does not appear the data accumulation.
In the $N$-layer HetMEC network, i.e., the HetMEC network with $N$ layers of MEC servers, the execution of tasks is related to the computing of the EDs, $N$ layers of MEC servers and the CC, as well as the transmitting of the EDs and $N$ layers of MEC servers.

\subsubsection{Constraints of the $n$-th Layer}
We consider the $n$-th layer\footnote{This layer may consist of the EDs or MEC servers. The constraints of them are similar.}, $1\leq n \leq N+1$. After receiving the raw data and the results, the nodes on the $n$-th layer need to process a part of the raw data, and transmit the rest raw data together with the processing results to the upper layer.
It is worth noting that the transmission resource allocated to each node on the $n$-th layer is determined by its parent node on the $(n-1)$-th layer.
We consider the case that all nodes on the $n$-th layer fully use their computing capacity, expressed by
\begin{align}
\lambda_{n}^{i} s_{n}^{i} = \theta_{n}^{i} = \theta_{n}^{i,u}, \ \forall 1\leq i \leq M_n, 1\leq n \leq N+1, \label{NBlockCon_n}
\end{align}
which implies that the transmission pressure of the $n$ layer is minimum.
Under the aforementioned circumstance, for each parent node $j$ on the $(n-1)$-th layer, the total volume of the data to transmit from the $n$-th layer to the parent node $j$ cannot surpasses the total transmission capacity of the parent node $j$.
Hence, when $n=N+1$ (the EDs), the constraints for transmitting are described in (\ref{NBlockCon_ED3}) and (\ref{NBlockCon_ED4}).
When $1\leq n\leq N$ (the MEC servers), the constraints for transmitting are described in (\ref{NBlockCon_AP3}) and (\ref{NBlockCon_AP4}).

\subsubsection{Constraints of the CC}
The CC needs to process all the remaining raw data delivered from the lower layer, and does not need to transmit. The volume of the arrived raw data at the CC should not surpass its computing capacity.
Hence, the constraint of the CC is expressed by
\begin{align}
\lambda_{0}^{1} \leq \theta_{0}^{1}.  \label{NBlockCon_CC}
\end{align}

Summarizing the constraints for all the layers of the HetMEC network, we have Proposition~\ref{prop:non-congested} to clarity the system constraints.
\begin{prop}\label{prop:non-congested}
The system constraints of the HetMEC network can be described by the constraints of each layer in the HetMEC network, i.e., the constraints that (\ref{NBlockCon_ED1}), (\ref{NBlockCon_ED3}), (\ref{NBlockCon_ED4}), (\ref{NBlockCon_AP1}), (\ref{NBlockCon_AP3}), (\ref{NBlockCon_AP4}), (\ref{NBlockCon_n}) and (\ref{NBlockCon_CC}) are all satisfied.
\end{prop}

\subsection{Latency Minimization Problem Formulation}
We aim to minimize the system latency of the HetMEC network, which is a general objective in the MEC networks. We then define the system latency in the HetMEC network and formulate the latency minimization problem in this subsection.

%In the HetMEC network, the data are processed and transmitted on each layer in a \emph{pipeline} mechanism. Therefore, the ED can start processing the new arrival task once it finishes processing its own part of the previous task, without waiting the previous task being processed completely and transmitted to the CC.

The latency of a task is defined as the sum of the computing time and transmitting time from the ED to the CC.
Consider the $n$-th layer where nodes receive the tasks from the lower layer.
The nodes need to process the assigned raw data from the lower layer, and deliver the processing results as well as unprocessed data to the upper layer. Therefore, the total latency of all nodes on the $n$ layer can be expressed by
\begin{equation}\label{latency_layer_n}
  L_n = \sum_{j=1}^{M_{n-1}}\sum_{i\in\mathcal{Q}_{n-1}^{j}}\left[ \frac{s_n^i \lambda_n^i}{\theta_n^i} + \frac{\rho s_n^i \lambda_n^i + (1-s_n^i)\lambda_n^i + \beta_n^i}{\phi_{n}^{j,i}} \right],
\end{equation}
where ${s_n^i \lambda_n^i}/{\theta_n^i}$ represents the processing time for the offloaded raw data, and $\\{[\rho s_n^i \lambda_n^i + (1-s_n^i)\lambda_n^i + \beta_n^i]}/{\phi_{n}^{j,i}}$ implies the transmitting time of the node $i$ on the $n$-th layer.

We then define the system latency as below.
\begin{Def}
The \textbf{system latency} is the total latency of all tasks generated at the EDs per unit time, the latency of one task including the computing time and transmission time of the ED, all layers of MEC servers and the CC.
\end{Def}
Therefore, the system latency can be expressed as
\begin{equation}\label{TotalLatency}
L = \frac{\lambda_0^1}{\theta_0^1} + \sum_{n=1}^{N+1}L_n,
\end{equation}
where ${\lambda_0^1}/{\theta_0^1}$ represents the computing time of the CC.

\par Hence, the total latency minimization problem in the HetMEC network can be formulated as below.
\begin{align}\label{Latency_min}
\min_{\textbf{s},\bm{\theta},\bm{\phi}} &\quad L, \\
s.t.\ &%(\ref{NBlockCon_ED1})-(\ref{NBlockCon_ED4}), (\ref{NBlockCon_AP1})-(\ref{NBlockCon_AP5}), (\ref{NBlockCon_CC1})-(\ref{NBlockCon_CC2}).\nonumber
(\ref{NBlockCon_ED1}), (\ref{NBlockCon_ED3}), (\ref{NBlockCon_ED4}), (\ref{NBlockCon_AP1}), (\ref{NBlockCon_AP3}), (\ref{NBlockCon_AP4}), (\ref{NBlockCon_n}), (\ref{NBlockCon_CC}). \nonumber
\end{align}

\iffalse
%%%%%%%%%%%%%%%%%%%%%%%%%%%%%%%%%%%%%%
\section{Performance Analysis}\label{sec:performance_analysis}
%%%%%%%%%%%%%%%%%%%%%%%%%%%%%%%%%%%%%%
In this section, we analyze how the number of layers influence the processing capacity and robustness of the whole HetMEC network.

\subsection{Why We Minimize the System Latency Only in the Non-congested HetMEC network}
We aim to minimize the system latency only when the HetMEC network is in the non-congested state, since the system latency in the congestion state is meaningless, because the data are no longer the real time data.

When the HetMEC network blocks, the historical data has already accumulated in the buffer, and thus, the new generated data cannot be processed until accumulated data is processed. The waiting time of the tasks in the buffer is much longer than the processing time. Therefore, the primary target in the congestion HetMEC network is to clear the accumulated data in the buffers as soon as possible. The optimization problem and corresponding algorithms in congestion state has been investigated in~\cite{EdgeFlow_iot}.
\fi

%\subsection{Influence of the Number of Layers to the Network Processing Capacity and Robustness}

%%%%%%%%%%%%%%%%%%%%%%%%%%%%%
\section{Latency Minimization Algorithm Design}\label{sec:algorithm_design}%
%%%%%%%%%%%%%%%%%%%%%%%%%%%%%
In this section, we propose a latency minimization algorithm (LMA) to solve problem (\ref{Latency_min}) via joint task assignment, computing and transmission resource allocation. We then analyze the influence of the number of MEC layers on the network robustness.

\subsection{Algorithm Design}
The system latency minimization problem described in (\ref{Latency_min}) is nonconvex, in which the task assignment strategy $\textbf{s}$, computing capacity allocation $\bm{\theta}$ and transmission resources allocation $\bm{\phi}$ are coupled.
By utilizing the Cauchy-Schwarz inequality \cite{CauchyInequality}, we can obtain the following inequations.
\begin{align}\label{CauchyInequation}
&L(\bm{s},\bm{\theta},\bm{\phi})=\frac{\lambda_0^1}{\theta_0^1} + \sum_{n=1}^{N+1}\sum_{j=1}^{M_{n-1}}\!\sum_{i\in\mathcal{Q}_{n-1}^{j}}\left[ \frac{s_n^i \lambda_n^i}{\theta_n^i} + \frac{\rho s_n^i \lambda_n^i + (1-s_n^i)\lambda_n^i + \beta_n^i}{\phi_{n}^{j,i}} \right]\nonumber \\
\geq &L_{min}(\bm{s}) \!=\!
 \left[\frac{\lambda_0^1}{\theta_{0}^{1,u}}\!+\!\sum_{n=1}^{N+1}\sum_{j=1}^{M_{n-1}}\!
 \sum_{i\in\mathcal{Q}_{n-1}^j}\frac{s_{n}^{i} \lambda_n^i}{\theta_{n}^{i,u}}\right] \!+\!
 \sum_{n=1}^{N+1}\sum_{j=1}^{M_{n-1}}\!\frac{\left(\sum_{i\in\mathcal{Q}_{n-1}^{j}}\!\!\sqrt{ \rho s_n^i \lambda_n^i + (1-s_n^i)\lambda_n^i + \beta_n^i } \right)^2}{\phi_{n-1}^j},
\end{align}
where $\theta_{n}^{i,u}$ and $\phi_{n-1}^j$ are the boundary of the computing and transmitting capacity.

\begin{prop}\label{Seperation}
 The task assignment strategy and resource allocation optimization can be separated in the proportional optimization problem (\ref{Latency_min}) by utilizing the Cauchy-Schwarz inequality.
\end{prop}

\begin{proof}
See Appendix \ref{app:seperation}.
\end{proof}

\begin{prop}\label{CauchyEqua}
The computing capacity division and transmission resources allocation can be derived by the following equations once the task assignment percentage is determined.
\begin{align}\label{RatioEquation}
&\frac{\phi_{n}^{j,i}}{\phi_{n}^{j,i'}} = \frac{\sqrt{\rho s_n^i \lambda_n^i + (1-s_n^i)\lambda_n^i + \beta_n^i}}{\sqrt{\rho s_n^{i'} \lambda_n^{i'} + (1-s_n^{i'})\lambda_n^{i'} + \beta_n^{i'}}}, \\
&\theta_{0}^{1} = \theta_{0}^{1,u}, \theta_{n}^{i} = \theta_{n}^{i,u}, \ \forall 1\leq n \leq N+1, 1\leq j \leq M_{n}, i,i'\in \mathcal{Q}_{n-1}^{j}.\nonumber
\end{align}
\end{prop}
\begin{proof}
According to (\ref{CauchyInequation}), the system latency $L(\bm{s},\bm{\theta},\bm{\phi})$ achieves a minimum $L_{min}(\bm{s})$ when the equality holds.
Based on the Cauchy-Schwarz inequality \cite{CauchyInequality}, the equality holds if and only if the equations in  (\ref{RatioEquation}) are satisfied, implying that the computing capacity division $\bm{\theta}$ and transmission resources allocation $\bm{\phi}$ can be derived from the task assignment percentage~$\bm{s}$.
\iffalse
\begin{align}\label{RatioEquation}
&\frac{\phi_{n}^{j,i}}{\phi_{n}^{j,i'}} = \frac{\sqrt{\rho s_n^i \lambda_n^i + (1-s_n^i)\lambda_n^i + \beta_n^i}}{\sqrt{\rho s_n^{i'} \lambda_n^{i'} + (1-s_n^{i'})\lambda_n^{i'} + \beta_n^{i'}}}, \\
&\theta_{1}^{1} = \theta_{1}^{1,u}, \theta_{n}^{i} = \theta_{n}^{i,u}, \ \forall 2\leq n \leq N, 1\leq j \leq M_{n}, i,i'\in \mathcal{Q}_{n-1}^{j}.\nonumber
\end{align}
\fi
\end{proof}

\begin{prop}\label{prop:Concave}
The task assignment problem (\ref{LatencyConvert}) is concave, and the optimal results $\bm{s}^*$ are at the vertex of the feasible set bounded by the constraints.
\begin{align}\label{LatencyConvert}
\min_{\bm{s}} &\quad L_{min}(\bm{s}), \\
s.t.\ &(\ref{NBlockCon_ED1}), (\ref{NBlockCon_ED3}), (\ref{NBlockCon_ED4}), (\ref{NBlockCon_AP1}), (\ref{NBlockCon_AP3}), (\ref{NBlockCon_AP4}), (\ref{NBlockCon_n}), (\ref{NBlockCon_CC}), (\ref{RatioEquation}). \nonumber
\end{align}
\end{prop}
\begin{proof}
See Appendix \ref{app:Concave}
\end{proof}

According to the aforementioned propositions, we can obtain the optimal results of the latency minimization problem by searching all the vertexes of the feasible set bounded by the constraints. The latency minimization algorithm is summarized in \textbf{Algorithm 1}. %\ref{alg:latencymin}.

\begin{algorithm}[t]
%\footnotesize
%\small
\label{alg:latencymin}
\caption{Latency minimization algorithm} % 算法的名字
\hspace*{0.02in} {\bf Input:} %算法的输入， \hspace*{0.02in}用来控制位置，同时利用 \\ 进行换行
  Computing capacity $\theta_{n}^{i}$, %\\
  upper bound of the transmission resource of each node $\phi_{n}^{i,0}$, %\\
  data generation speed $\bm{\lambda}$.\\
\hspace*{0.02in} {\bf Output:} %算法的结果输出
Task assignment strategy $\bm{s}^*$, resources allocation scheme $\bm{\theta}^*, \bm{\phi}^*$.
\begin{algorithmic}[1]
%\State Convert the proportional optimization problem in (\ref{Latency_min}) into the of task assignment problem in (\ref{CauchyInequation}) by utilizing Cauchy-Schwarz inequality.
\ForAll{Vertex of the feasible set}
    \State Obtain the corresponding task assignment strategy $\bm{s}$.
    \State Obtain the resource allocation scheme $\bm{\theta}, \bm{\phi}$ according to $\bm{s}$ and (\ref{RatioEquation}).
    \If{ Non-congested constraints are satisfied }
        \If{$L_{min}(\bm{s}) < L_{min}(\bm{s}^*)$}
            \State $L_{min}(\bm{s}^*) = L_{min}(\bm{s})$.
            \State Update the optimal $\bm{s}^*, \bm{\theta}^*$ and $\bm{\phi}^*$.
        \EndIf
    \EndIf
\EndFor
\end{algorithmic}
\end{algorithm}

We first convert the proportional optimization problem in (\ref{Latency_min}) into the of task assignment problem in (\ref{CauchyInequation}) by utilizing Cauchy-Schwarz inequality. Since the converted task assignment problem is proved concave according to \textbf{Proposition \ref{prop:Concave}}, we calculate the system latency at each vertex of the feasible set bounded by the non-congested constraints, i.e., $L_{min}(\bm{s})$, where $\bm{s}$ is determined by the constraints associated with the discussing vertex. The computing and transmission resource allocation are also determined once the task assignment strategy $\bm{s}$ is fixed. After considering all the vertexes, the minimum system latency together with the optimal task assignment strategy and resource allocation can be obtained.

\begin{remark}
Assuming the number of nodes in the whole network is denoted by $M$,
the complexity of the latency minimization algorithm is proportional to the number of the vertexes of the feasible set. The maximum number of the vertexes of the feasible set is in the square magnitude of $M$, i.e., $O(M^2)$.
\end{remark}
\begin{proof}
See Appendix \ref{app:complexity}.
\end{proof}

\subsection{Network Robustness Analysis}
In this subsection, we discuss the relation between the network robustness and the number of MEC layers of the HetMEC network, influenced by the amount of computing and transmission resources.

The network robustness is reflected by the network processing capacity, which is defined by the maximum data generation speed at the ED supported by the non-congested HetMEC networks.
The ED can start processing the new arrival task once it finishes processing its own part of the previous task, without waiting it to be processed completely and transmitted to the CC.
Intuitively, a network is more robust when more resources are available brought by a newly added layer. However, this may only be true when the added layer is selected properly, as will be analyzed in both computing and transmission resource shortage cases as below.

\subsubsection{Computing Resource Shortage Case}

\begin{figure}[t]
\centering
\includegraphics[width=0.95\textwidth]{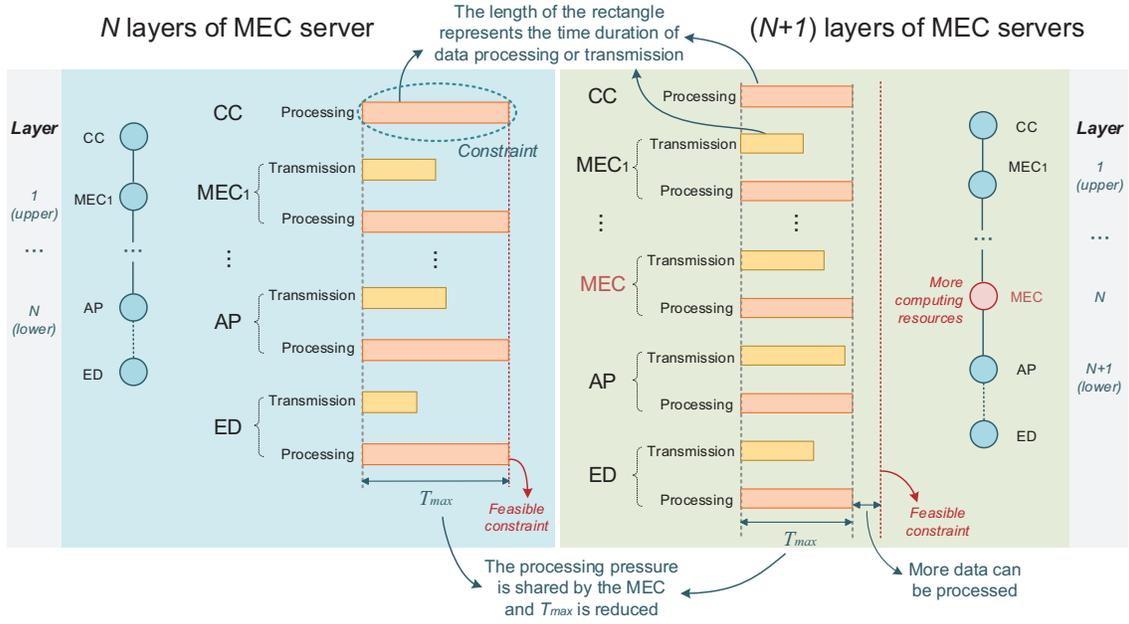}
\caption{The computing resource shortage case.}\label{fig:casestudy1}
\end{figure}

As shown in Fig.~\ref{fig:casestudy1},
the bottleneck of the network processing capacity lies in the limited computing resources, which can be expressed by
\begin{align}\label{Computing limited}
  & \lambda_n^is_n^i = \theta_n^{i,u}, \forall 1\leq n\leq N, 1\leq i\leq M_n, \\
  & \sum_{j\in\mathcal{Q}_{n-1}^{k}}\left(\rho\lambda_{n}^{j} s_{n}^{j}+\lambda_{n}^{j}(1 - s_{n}^{j}) + \beta_{n}^{j}\right) < \phi_{n-1}^{k}, \forall 2\leq n\leq N, 1\leq k\leq M_{n-1}.
\end{align}
In this case, the computing resources of all layers are fully utilized, while there still remains the idle transmission resources in the HetMEC network. The network will be in congestion if the data generation speed $\lambda$ continues to increase.

When adding a layer of MEC servers between the $(n_0-1)$-th and $n_0$-th layer of the initial network\footnote{The added layer becomes the $n_0$-th layer, and the initial $n_0$-th layer becomes the $(n_0+1)$-th layer.}, in order to increase the network robustness, the computing and transmission resources of the new added layer should satisfy that
\begin{align}\label{Computing_addcondition}
  & \theta_{n_0}^{i,u}>0, \forall 1\leq i\leq M_{n_0}, \\
  & \sum_{j\in\mathcal{Q}_{n_0}^{i}}\left(\rho\lambda_{n_0+1}^{j} s_{n_0+1}^{j}+\lambda_{n_0+1}^{j}(1 - s_{n_0+1}^{j}) + \beta_{n_0+1}^{j}\right) < \phi_{n_0}^{i}, \forall 1\leq i\leq M_{n_0},
\end{align}
where each node on the added layer possesses the computing resources, and the transmission resources of each node on the added layer are sufficient enough to transmit all the data from the lower layer (i.e., $(n_0+1)$-th layer in the $(N+1)$-layer network). Therefore, the added layer can relieve the processing pressure of the other layers.
As the data generation speed $\lambda$ continues to increase, the task division percentage $s_{n}^i, 1\leq i\leq M_{n},$ on any other layer, $n\neq n_0,$ can be reduced, and the task division percentage $s_{n_0}^i, 1\leq i\leq M_{n_0},$ increases until the following conditions are satisfied.
\begin{align}
  &\frac{\lambda_{n_0}^i s_{n_0}^i}{\theta_{n_0}^{i}} = \frac{\lambda_n^j s_n^j}{\theta_n^{j}}, \forall n\neq n_0, 1\leq i\leq M_{n_0}, 1\leq j\leq M_{n},  \\
  &\lambda_{n_0}^i s_{n_0}^i \leq \theta_{n_0}^{i,u}, \forall 1\leq i\leq M_{n_0}, \\
  & \sum_{j\in\mathcal{Q}_{n-1}^{k}}\left(\rho\lambda_{n}^{j} s_{n}^{j}+\lambda_{n}^{j}(1 - s_{n}^{j}) + \beta_{n}^{j}\right) < \phi_{n-1}^{k}, \forall 2\leq n\leq N+1, 1\leq k\leq M_{n-1}.
\end{align}
The computing resources of the whole HetMEC network are then fully utilized where the processing time on all layers are equal and limited by the non-congested computing constraint.

%As illustrated in Fig. \ref{fig:casestudy1}, the computing resources constrain the network processing capacity. The computing resources of the MEC are introduced into the HetMEC network when the number of layers increases. Other layers, including the ED, AP and CC layer, can offload part of raw data to the MEC layer, and the processing time of these layers can be reduced. Therefore, the HetMEC network can process the tasks with higher arrival speed, showing larger network processing capacity and stronger robustness.

\subsubsection{Transmission Resource Shortage Case}
\begin{figure}[t]
\centering
\includegraphics[width=0.95\textwidth]{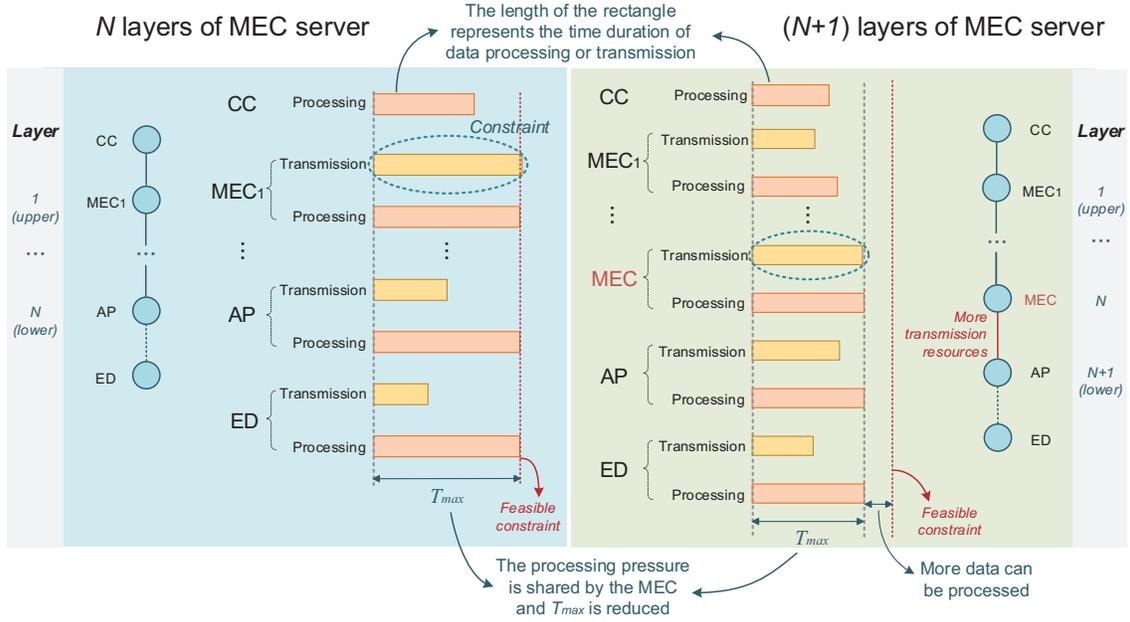}
\caption{The transmission resource shortage case.}\label{fig:casestudy2}
\end{figure}

As shown in Fig. \ref{fig:casestudy2}, The bottleneck of the the network processing capacity lies in the transmission resources of one layer $n_0$, and the layer is determined by the conditions described in (4), (5), (12) and (13).
The transmission resources of the $n_0$-th layer have been fully utilized for the data transmission from the $(n_0+1)$-th layer to the $n_0$-th layer, which can be expressed by
\begin{align}\label{Transmission limited}
  & \lambda_n^is_n^i < \theta_n^{i,u}, \forall 0\leq n\leq N+1, 1\leq i\leq M_n, \\
  & \sum_{j\in\mathcal{Q}_{n_0}^{k}}\left(\rho\lambda_{n_0+1}^{j} s_{n_0+1}^{j}+\lambda_{n_0+1}^{j}(1 - s_{n_0+1}^{j}) + \beta_{n_0+1}^{j}\right) = \phi_{n_0}^{k}, \forall 1\leq k\leq M_{n_0+1}.
\end{align}
The network will be congested between the $(n_0+1)$-th layer and the $n_0$-th layer if the data generation speed $\lambda$ continues to increase.

The robustness can be enhanced only when adding a layer of MEC servers between the $(n_0+1)$-th layer and the $n_0$-th layer, or below the $(n_0+1)$-th layer.
It does not make any contribution to the robustness enhancement to add a layer when the layer number is smaller than $n_0$ (i.e., the added layer is above the $n_0$-th layer). This is because the transmission resources are allocated by the parent node on the upper adjacent layer, and the operation of adding a layer of MEC servers above the $n_0$-th layer cannot increase the transmission resources of the $n_0$-th layer or reduce the amount of data transmitted from the $(n_0+1)$-th layer.

When adding a layer of MEC servers below the initial $n_0$-th layer, denoted by the $n'$-th layer\footnote{The added layer becomes the $n'$-th layer, and the initial $n'$-th layer becomes the $(n'+1)$-th layer, and $n'>n_0$.}, in order to increase the network robustness, the computing and  transmission resources of the added layer should satisfy that
\begin{align}\label{Transmission_addcondition}
  & \theta_{n'}^{i,u}>0, \forall 1\leq i\leq M_{n'}, \\
  & \sum_{j\in\mathcal{Q}_{n'}^{i}}\left(\rho\lambda_{n'+1}^{j} s_{n'+1}^{j}+\lambda_{n'+1}^{j}(1 - s_{n'+1}^{j}) + \beta_{n'+1}^{j}\right) < \phi_{n'}^{i}, \forall 1\leq i\leq M_{n'},
\end{align}
where each node on the added layer possesses the computing resources, and the transmission
resources are sufficient enough to transmit all the data from the lower layer.

Tasks with a larger data generation speed can be processed since condition (31) has changed to the condition that the amount of the transmitted data from the $(n_0+1)$-th layer is strictly smaller than the transmission resources. Therefore, the network can remain non-congested when the data generation speed $\lambda$ continues to increase. Condition (31) changes in two cases:
\begin{itemize}
  \item If a layer is added between the $(n_0+1)$-th layer and the $n_0$-th layer, condition (31) changes because the transmission resources are more abundant\footnote{This is to say, the right side of (31) becomes larger.}.
  Therefore, the amount of data to be transmitted from the $(n_0+1)$-th layer satisfies that \\ $\sum_{j\in\mathcal{Q}_{n_0}^{k}}\left(\rho\lambda_{n_0+1}^{j} s_{n_0+1}^{j}+\lambda_{n_0+1}^{j}(1 - s_{n_0+1}^{j}) + \beta_{n_0+1}^{j}\right) < \phi_{n_0+1}^k.$
  \item If a layer is added below the $(n_0+1)$-th layer, condition (31) changes because the amount of data to transmit is reduced\footnote{This is to say, the left side of (31) becomes smaller.}.
  We assume that $\lambda_0$ volume of data that originally being processed on the $(n_0+1)$-th layer are offloaded to the added layer.  The amount of data to be transmitted from the $(n_0+1)$-th layer is reduced and satisfies that \\ $\sum_{j\in\mathcal{Q}_{n_0}^{k}}\left(\rho\lambda_{n_0+1}^{j} s_{n_0+1}^{j}+\lambda_{n_0+1}^{j}(1 - s_{n_0+1}^{j}) + \beta_{n_0+1}^{j} - (1-\rho)(1-s_{n_0+1}^{j})\lambda_0\right) < \phi_{n_0+1}^k$.
\end{itemize}

We summarize the relation between the network robustness and the number of MEC layers influenced by the amount of computing and transmission resources.
In the computing resource shortage case, the network robustness can be enhanced if the MEC servers which are enabled the computing capacity on the new added layer satisfies (\ref{Computing_addcondition}) and (26).
In the transmission resource shortage case, the network robustness can be enhanced only if a layer of MEC servers satisfying (\ref{Transmission_addcondition}) and (33) are added \emph{below} the initial transmission resource constrained MEC layer.

%As illustrated in Fig. \ref{fig:casestudy1}, the transmission resources of the AP constrains the network processing capacity. The EDs and AP cannot offload raw data to the CC, which will result in heavier transmission burden at the AP.
%When adding a layer of the MEC server, the MEC server holds more transmission resources than the CC, the transmission time of the AP can be reduced, and the ED layer and AP layer can offload more raw data to the MEC layer and CC layer.
%Therefore, the HetMEC network can process tasks with higher arriving speed, showing larger network processing capacity and stronger robustness.

%%%%%%%%%%%%%%%%%%%%%%%%%%%%%
\section{Simulation Results}\label{sec:simulation_results}%
%%%%%%%%%%%%%%%%%%%%%%%%%%%%%
In this section, we evaluate the system latency and the processing rate in the HetMEC network performing our algorithm LMA and other task assignment schemes. The robustness in the HetMEC network with different number of layers is also investigated.

\subsection{Parameters Setting}
In our simulation, the parameters about the data processing are presented in Table \ref{specifications}. The EDs transmit the file of size $60$ Kbits in a period of $1$s. The compression ratio of the raw data after processing is set as $10\%$. The computing capacity is represented by the maximum volume of the processed data per second, and the transmission resources are reflected by the transmission bandwidth.
\begin{table}[t]
	\centering
	\caption{HetMEC network Parameters}
\vspace{-3mm}
	\begin{tabular}{l|l}
		\hline \hline
        Parameters & Value \\
        \hline \hline
%		The CPU frequency of the ED device  & $1 * 10^9$ Hz\\
%		\hline
%		The CPU frequency of the MEC server or AP device & $3 * 10^9$ Hz\\
%		\hline
%		The CPU frequency of the CC device & $1 * 10^{10}$ Hz\\
%		\hline
%		Wireless transmission resources of each AP & 300 Kbps \\
%		\hline
%		Wired transmission resources & 1 Mbps\\
%		\hline
		The volume of the data file & 60 Kbits\\
		\hline
		Compression ratio $\rho$ & $10\%$ \\
		\hline
		The period of data generation & 1s \\
		\hline \hline
	\end{tabular} \label{specifications}
\end{table}
\begin{table}[t]
\small
	\centering
	\caption{Network architecture of the HetMEC network with different number of MEC layers.}
\vspace{-3mm}
	\begin{tabular}{l|c|c|c|c}
		\hline \hline
		      & Cloud-only & One-layer HetMEC& Two-layer HetMEC& Three-layer HetMEC\\
		\hline \hline
		 Consideration of the EDs      & \checkmark  & \checkmark & \checkmark & \checkmark\\
		\hline
		\multirow{2}*{Number of the MEC layer}   & \multirow{2}*{0} & 1 & 2& 3 (AP, switch,\\
        ~&~& (AP)& (AP, switch) &network gateway)\\
		\hline
		{Consideration of the CC}     & \checkmark & \checkmark & \checkmark & \checkmark\\
		\hline \hline
	\end{tabular} \label{Table:architecture}
\end{table}
We consider four cases in our simulation, i.e., the cloud-only network, 1-layer, 2-layer and 3-layer HetMEC networks, as shown in Table \ref{Table:architecture}. The computing capacity and transmission resource settings of the three cases are presented in Table \ref{Compcases}.
In all cases, the data generated at the EDs need to be aggregated at the CC through the MEC servers, while the task assignment strategy is performed among the EDs, all layers of MEC servers and the CC.
\begin{table}[t]
	\centering
	\caption{Computing capacity and transmission resource setting of the HetMEC network}
\vspace{-3mm}
	\begin{tabular}{c|c|c}
		\hline \hline
		Node     & Computing capacity & Transmission resources\\
		\hline \hline
		 ED       & $0.12$ Mbps &-\\
		\hline
		AP        & $0.4$ Mbps & $1.2$ Mbps\\
		\hline
		Switch (Lower-layer MEC server)     & $1.5$Mbps & $3$ Mbps\\
		\hline
		Network gateway (Upper-layer MEC server)          & $4.2$ Mbps & $4.8$ Mbps\\
		\hline
		CC          & $12$ Mbps & $12$ Mbps\\
		\hline \hline
	\end{tabular} \label{Compcases}
\end{table}

\iffalse
\subsection{Task Assignment Strategy Analysis}
\begin{figure}[t]
	\centering
	\includegraphics[width=0.7\textwidth]{TAresult.eps}
	\caption{The task division percentage of each ED in the two-layer HetMEC network.}
	\label{Fig:TAresult}
\end{figure}
The division and assignment strategies of the tasks at each ED and the CC at different data generation speed in the two-layer cases are shown in Fig. \ref{Fig:TAresult} . With the low data generation speed, i.e., $\lambda\leq3$, all raw data are offloaded to the CC for processing in order to achieve the minimum system latency. When the data generation speed continues to increase, the CC cannot process the offloaded data in real time if the EDs still offload all data to the CC, and thus, the EDs process part of the data locally and offload the rest raw data to the CC for processing. Though the latency is larger for the EDs to process the data locally, the network can remain non-congested.

\fi

\subsection{Network Performance Evaluation}
The network performance is evaluated based on the following metrics:
\begin{itemize}
	\item{\textbf{System Latency}: The total latency of all tasks generated at the EDs per second.}
%The system latency is the sum time interval of all tasks between the generation of the task and the terminal of the processing time at the CC.}
	\item{\textbf{Processing Rate}: The average volume of data processed by the HetMEC network per second viewed by each ED.}
	\item{\textbf{Network Robustness}: The network robustness is represented by the maximum data generation speed supported by the non-congested HetMEC network.}%the average number of the unprocessed packages in the network, which represents the degree of blocking in the system. }
\end{itemize}

We compare the LMA with the following task assignment schemes on these metrics.
\begin{itemize}
  \item \emph{Cloud computing}: The cloud computing scheme indicates that all the data are processed at the CC.
  \item \emph{Local computing}: The local computing scheme indicates that all the data are processed at the ED without being offloaded to the MEC servers or CC.
  \item \emph{Conventional MEC}: The conventional MEC scheme indicates that the data are totally offloaded to the APs for processing.
\end{itemize}

\subsubsection{System Latency Evaluation}

\begin{figure}[t]
	\centering
%	\subfigure[Two-layer HetMEC network]{\includegraphics[width=0.6\textwidth]{latency_2layer.eps}}
	\subfigure[One-layer HetMEC network]{\includegraphics[width=0.49\textwidth]{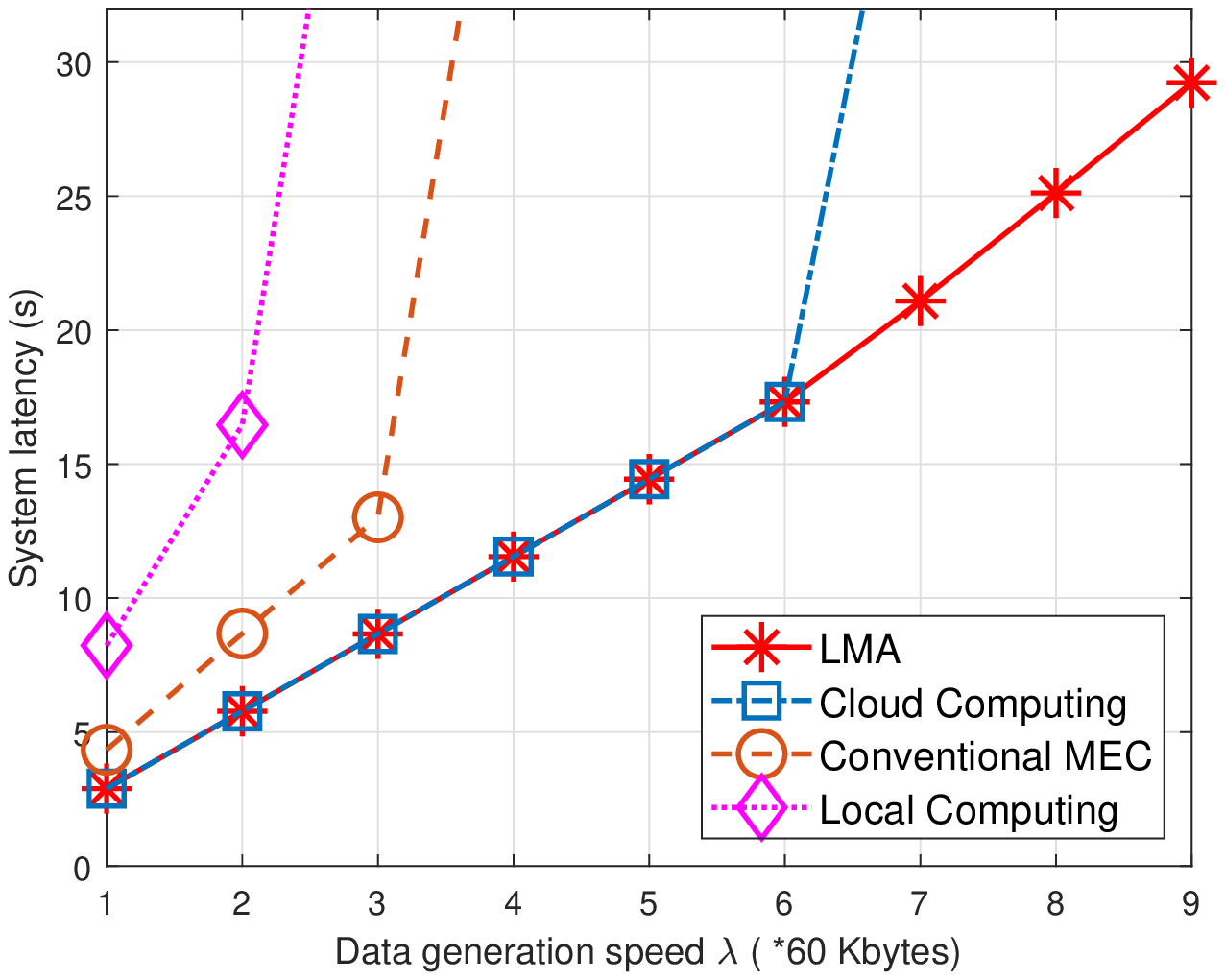}}
	\subfigure[Two-layer HetMEC network]{\includegraphics[width=0.49\textwidth]{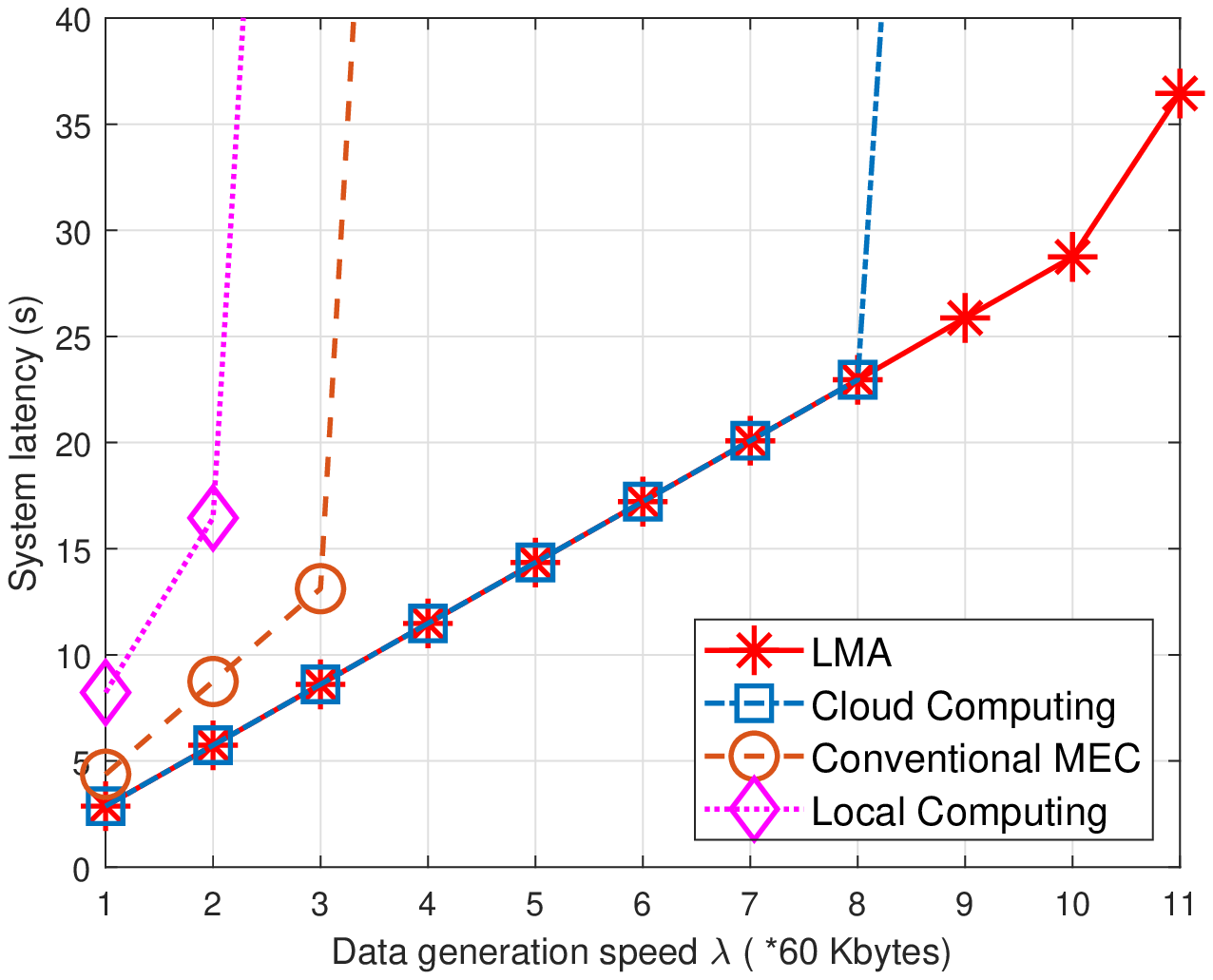}}

%	\subfigure[Compare in the HetMEC network with different layers]{\includegraphics[width=0.5\textwidth]{latency_compare.eps}}
	\caption{ The system latency vs. the data generation speed in the different cases.}
	\label{fig:latency}
\end{figure}

\par Fig. \ref{fig:latency} presents the system latency versus the data generation speed $\lambda$ at the ED in different cases.
In the one-layer HetMEC network, as shown in Fig. \ref{fig:latency}(a), the system latency of all schemes increases with the data generation speed, since more data need to be processed, resulting in larger system latency.
By performing our proposed algorithm LMA, the system latency remains the lowest given different data generation speed.
When the data generation speed $\lambda>6$, the slop of the line segment of LMA becomes larger, reflecting that the average latency increases. When the data generation speed is small, the latency is smallest when the data are processed at the uppermost layer.
When the data generation speed is large, the data need to be offloaded to other layers due to the limited computing capacity of the uppermost layer, which induces the increase of the average latency.

In the two-layer HetMEC network, as shown in Fig. \ref{fig:latency}(b), the system latency of the LMA also remains the lowest in different cases given different data generation speeds, which reflects the advantages of the LMA in the HetMEC network. Compared with other schemes, the LMA jointly utilizes the computing capacity and transmission resources of all devices, and thus, its latency remains low with the data generation speed increasing.
At the data generation speed $\lambda=11$, the new generated data can be processed in real time when performing our algorithm LMA, showing the robustness of our scheme.
Given the same data generation speed, the system latency of the two-layer HetMEC network performing the LMA is not larger than that of the one-layer HetMEC network, and the robustness of the two-layer HetMEC network performing the LMA is stronger than that of the one-layer HetMEC network.

%the differences between the system latency of the LMA and other schemes are larger given higher data generation speed, e.g., when $\lambda>3$. The number of devices, the available computing, and transmission resources are more in one-layer HetMEC network. Therefore, in the one-layer HetMEC network, the advantages of the LMA is more obvious when jointly utilizing the resources of the whole network.

%It is obvious that the system latency of our EdgeFlow system remains the lowest in different cases with different data generation speed, which reflects the advantages of the EdgeFlow system. As the data generation speed increases, the system needs more resources for the data processing and transmitting. Moreover, compared with the two-layer system, the performance of the one-layer EdgeFlow system is superior than other schemes more notably.

\subsubsection{Processing Rate Evaluation}

\begin{figure}[t]
	\centering
	\subfigure[One-layer HetMEC network]{\includegraphics[width=0.49\textwidth]{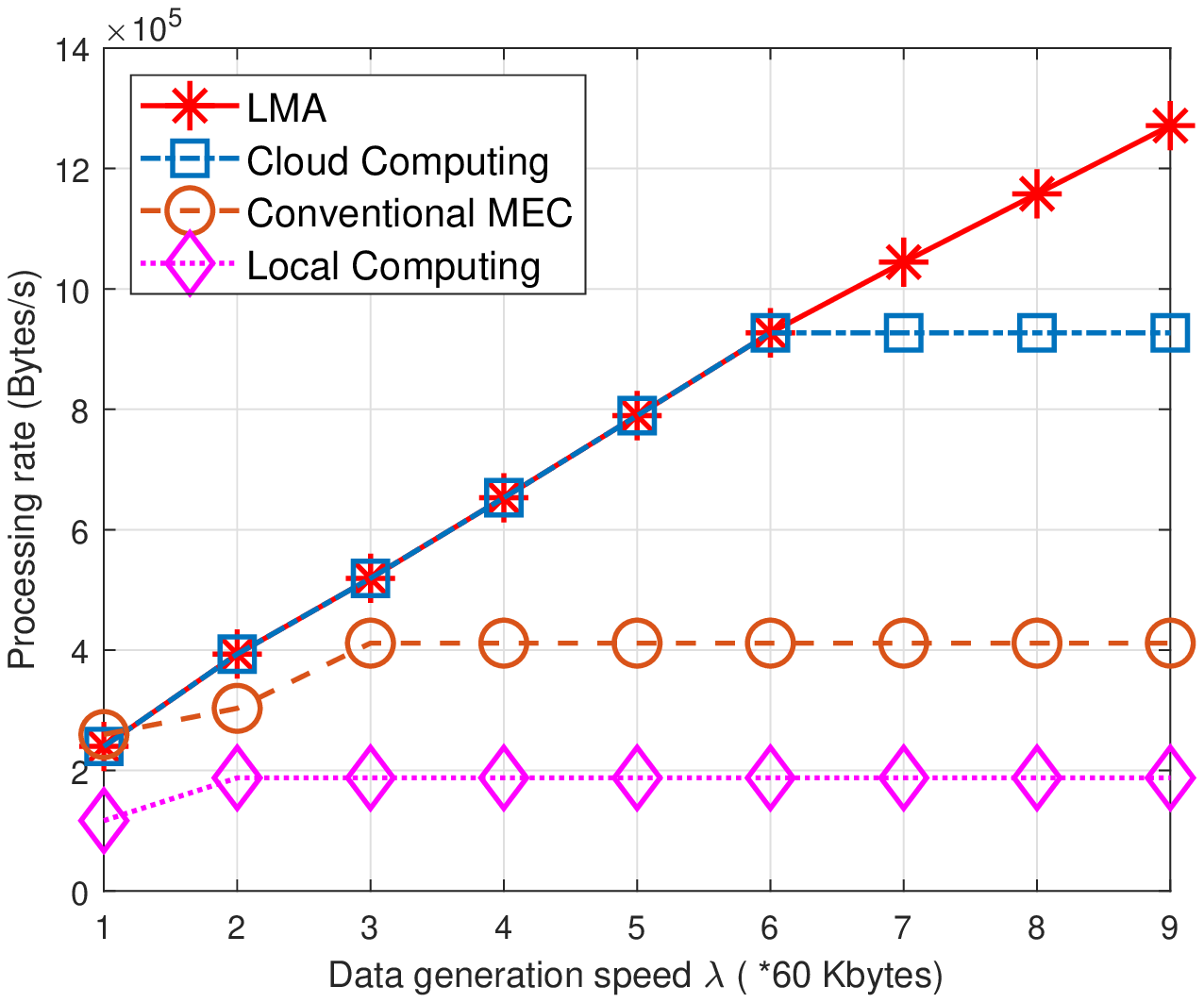}}
    \subfigure[Two-layer HetMEC network]{\includegraphics[width=0.49\textwidth]{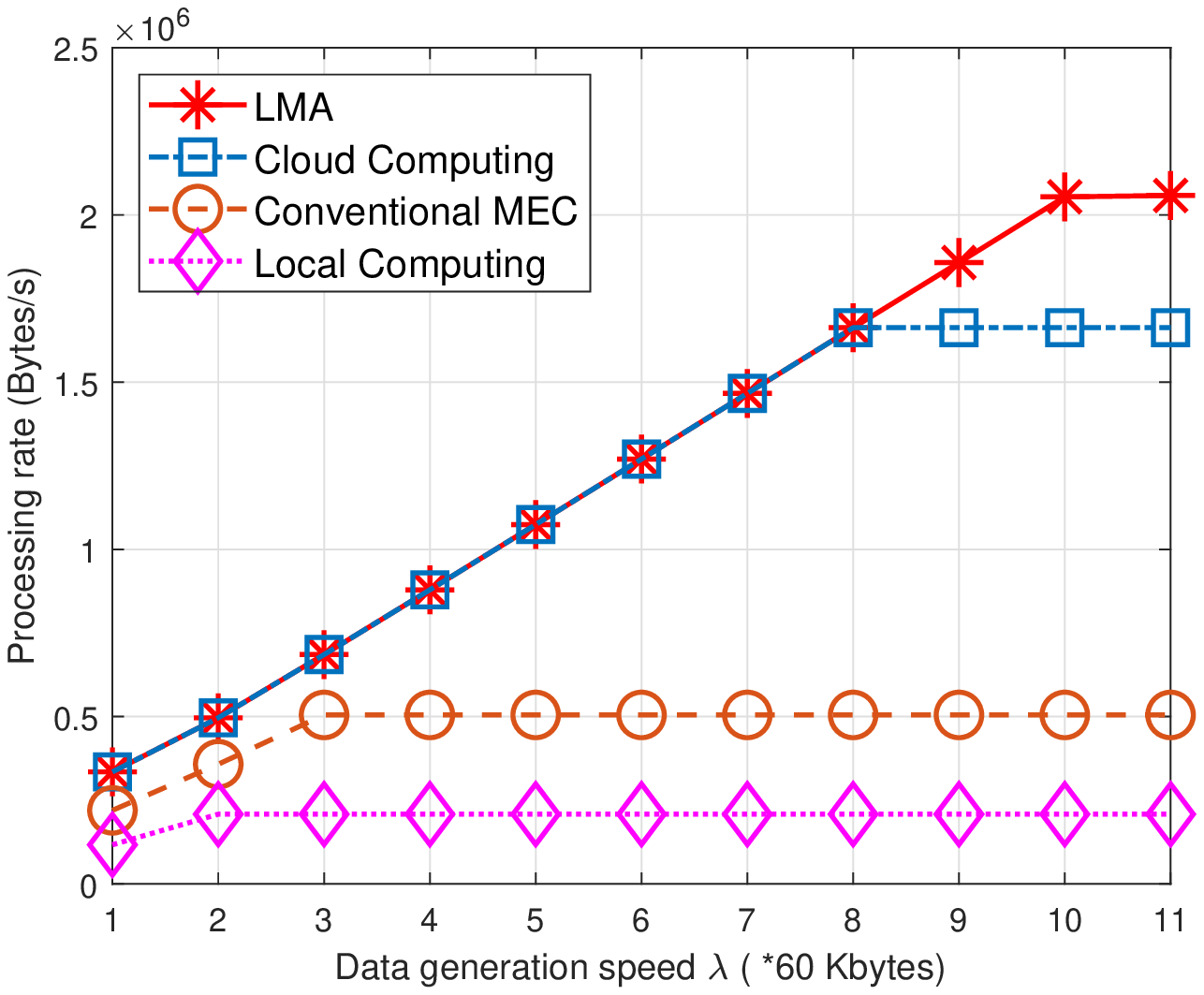}}
    %\subfigure[Compare of two-layer and one-layer EdgeFlow system]{\includegraphics[width=0.48\textwidth]{AverageProcRate_LayerCompare.eps}}
	\caption{The processing rate with the increase of the data generation speed in the HetMEC network.}
	\label{Fig:computational_throughput}
\end{figure}

\begin{figure}[t]
	\centering
    \includegraphics[width=0.5\textwidth]{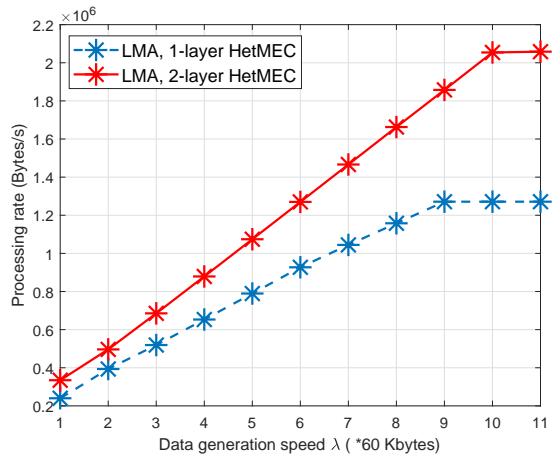}
    \caption{The processing rate vs. the data generation speed in the HetMEC network with different number of layers.}
	\label{Fig:throughput_compare}
\end{figure}

\par As shown in Fig. \ref{Fig:computational_throughput}, we analyze the processing rate given different data generation speed in both cases.
As presented in Fig. \ref{Fig:computational_throughput}(a) and (b), both in the one-layer HetMEC network and two-layer HetMEC network, the processing rate of different schemes is non-decreasing as the data generation speed increases. Other schemes, e.g., conventional MEC scheme, the cloud and local computing, reach the saturated point when the data generation speed surpasses the respective thresholds, reflecting the bottlenecks of the network processing rate utilizing these schemes. After reaching the bottleneck, the network cannot offer more computing resources for the new generated data, and the processing rate stops increasing. The LMA gains relatively high processing rate, especially when the data generation speed is large. Since the LMA jointly utilizes the computing and transmission resources of the whole HetMEC network, it achieves the highest bottleneck of the processing rate.

\par Fig. \ref{Fig:throughput_compare} compares the processing rate in the one-layer and two-layer HetMEC networks by performing LMA. When the HetMEC network is congested, no more resources can be utilized for data processing, and thus, the processing rate reaches saturation and stop grows with the data generation speed.
Compared with the one-layer HetMEC network, the processing rate of the two-layer HetMEC network is higher given the same data generation speed, especially when the data generation speed is large.
The computing and transmission resources are enriched in the two-layer HetMEC network. By jointly utilizing the computing resources of different layers and properly scheduling the data transmission among layers in the HetMEC network, more computing resources contributes to higher processing rate as the number of layers grows.

%Since there are more devices in the one-layer system, the available computing capacity is more than that of the two-layer system, resulting in higher processing rate.
%The processing rate of different schemes increases as the data generation speed becomes larger until saturation, which represents the bottleneck of utilizing the computing capacity of the whole system. Apparently, the EdgeFlow system gains relatively high processing rate given different data generation speed, especially when other schemes reach their bottleneck. Hence, the EdgeFlow system can utilize the computing capacity of the whole system more fully.

\subsubsection{Network Robustness Evaluation}

\begin{table}[t]
	\centering
    %\small
	\caption{Computing capacity and transmission resource setting in different cases}
\vspace{-3mm}
	\begin{tabular}{c|c|c|c|c|c|c}
		\hline \hline
		\multirow{2}*{Node}     & \multicolumn{2}{c|}{Case 1}& \multicolumn{2}{c|}{Case 2}& \multicolumn{2}{c}{Case 3}\\
        \cline{2-7}
		~     & Compute & Transmit & Compute & Transmit & Compute & Transmit \\
		\hline \hline
		 ED       & $0.12$ Mbps &- & $0.12$ Mbps &- & $0.12$ Mbps &-\\
		\hline
		AP        & $0.4$ Mbps & $0.9$ Mbps & $0.4$ Mbps & $1.2$ Mbps & $0.4$ Mbps & $3$ Mbps\\
		\hline
		Switch     & $1.5$Mbps & $3$ Mbps & $1.5$Mbps & $3$ Mbps & $1.5$Mbps & $6$ Mbps\\
		\hline
		Network            & \multirow{2}*{$4.2$ Mbps} & \multirow{2}*{$4.8$ Mbps} & \multirow{2}*{$4.2$ Mbps} & \multirow{2}*{$4.8$ Mbps} & \multirow{2}*{$4.2$ Mbps} & \multirow{2}*{$12$ Mbps}\\
		 gateway           &~ &~ &~ &~ &~ &~ \\
		\hline
		CC          & $12$ Mbps & $12$ Mbps & $12$ Mbps & $12$ Mbps & $12$ Mbps & $15$ Mbps\\
		\hline \hline
	\end{tabular} \label{Table:case}
\end{table}

\begin{figure}[t]
	\centering
	\includegraphics[width=0.55\textwidth]{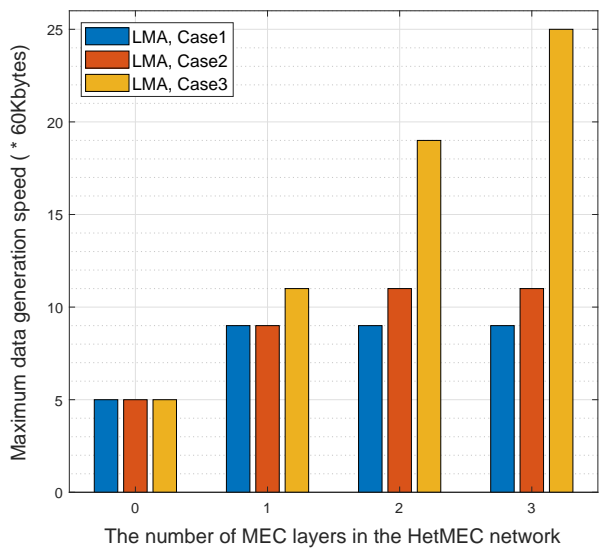}
	\caption{The network robustness with different number of MEC layers in the HetMEC network.}
	\label{Fig:bufferSize}
\end{figure}

Fig. \ref{Fig:bufferSize} shows the maximum data generation speed of the HetMEC network in the cloud-only network, one-layer, two-layer and three-layer HetMEC networks in different cases, the settings of which are presented in Table \ref{Table:case}.
The maximum data generation speed of the non-congested HetMEC network reflects the robustness of the network, as analyzed in Section IV. Case 1 and case 2 show the performance of robustness in the transmission resource shortage case of the HetMEC network, and case 3 shows the performance of robustness in the computing resource shortage case of the HetMEC network.
%Comparing different schemes in the HetMEC network in the same case, we observe that the LMA performs best among all schemes. Local computing scheme is constrained by the weak computing capacity of the EDs, while the conventional distributed scheme is constrained by the computing capacity of the APs. The cloud computing scheme only relies on the computing resources of the uppermost layer, and the computing resources of the middle layers are ignored.

In case 1, as the number of the MEC layers increases, the network robustness first becomes stronger when the number of MEC layers $N\leq1$, and then remains the same when~$N\geq 1$.
\emph{When $N\leq 1$}, the main constraint of the robustness is the computing resources in the whole network. Therefore, the robustness of the HetMEC network become stronger when inducing more layers of MEC servers into the task assignment and processing.
\emph{When $N\geq1$}, it is the transmission resources between the APs and EDs that constrains the network robustness. As analyzed in Section~IV, it cannot contribute to the robustness enhancement to add a layer of MEC servers above the APs, and thus, the network robustness does not increase anymore.
\iffalse
\begin{itemize}
\item \emph{When $N<1$}: The main constraint of the robustness is the computing resources in the whole network. Therefore, the robustness of the HetMEC network become stronger when inducing more layers of MEC servers into the task assignment and processing.
\item \emph{When $N\geq1$}: It is the transmission resources between the APs and EDs that constrains the network robustness. As analyzed in Section~IV, it cannot contribute to the robustness enhancement to add a layer of MEC servers above the APs, and thus, the network robustness does not increase anymore.
\end{itemize}
\fi

In case 2, as the number of MEC layers grows, the network robustness becomes stronger when the number of MEC layers $N\leq2$, and then remains the same when $N\geq 2$. Since the transmission resources between the APs and the EDs become more abundant, the HetMEC network can process the tasks with larger data generation speed, and thus, the robustness is enhanced in the two-layer HetMEC network. However, the transmission resources between the APs and the EDs still constrain the enhancement of the robustness, and the robustness remains unchanged when $N\geq 2$.
The network robustness is improved only when the computing and transmission resources of the added layer satisfy the conditions analyzed in Section IV.

In case 3, the network robustness becomes stronger as the number of MEC layers grows. The transmission resources of all MEC layers are abundant, and the computing resources constrain the improvement of the robustness.
As analyzed in Section IV, in the computing resource shortage case, it can improve the network robustness to induce new layers of MEC servers satisfying (\ref{Computing_addcondition}) and (26).

\vspace{5mm}
%%%%%%%%%%%%%%%%%%%%%%
\section{Conclusion \label{sec:conclusion}}%
%%%%%%%%%%%%%%%%%%%%%%

In this paper, we have studied a HetMEC network in order to provide low-latency data services.
We have considered a typical uplink MEC application, where the raw data are generated at the EDs and the results of the data processing need to be aggregated at the CC through multiple layers of MEC servers.
The tasks are optimally divided and assigned to the nodes on multiple layers, including the CC, MEC servers and EDs.
Through jointly considering the task assignment, computing and transmission resource allocation, we have proposed the LMA for latency minimization in the HetMEC network.
Simulation results have showed that our proposed algorithm LMA can significantly reduce the system latency and increase the processing rate as well as the network robustness.
Based on both theoretical and numerical analysis, we conclude that
the relation between the network robustness and the number of layers of the HetMEC network is influenced by the amount of the computing and transmission resources.
In the computing resource shortage case, the robustness can be improved when inducing the MEC servers above any layer.
In contract, in the transmission resource shortage case, it cannot contribute to the enhancement of the robustness when inducing more layers of MEC servers above the initial transmission resource constrained layer.

\iffalse

In this paper, we have proposed a multi-layer data flow processing system EdgeFlow, which consists of the CC, APs and the EDs.
The EdgeFlow system can provide the low latency services for the IoT real-time applications via the integrated utilization of the computing capacity and transmission resource of both cloud center and edge nodes.
The blocking and non-blocking states have been investigated and the quantitve boundary between the two states has been derived in Proposition~1. In the non-blocking state, the system latency is minimized, while in the blocking state, the latency is meaningless for the accumulated data and the recovery time of the system is minimized.
The multi-layer collaborative task assignment and resource allocation strategies have been proposed in Algorithms~1 and~2 for both states to achieve the optimal solutions.
The implementation of the EdgeFlow system is based on the USRPs, the Intel NUCs and the Linux system for the typical IoT applications, face recognition.
Experimental results have showed that our EdgeFlow system can obviously reduce the system latency and increase the data processing rate, especially in the case of high data generation speed.
The system can stay in the non-blocking state by the dynamic task assignment strategy and the resources allocation when the data generation speed increases, and thus the volume of accumulated data in the buffer remains small.
\fi

\vspace{10mm}
\begin{appendices}
\section{Proof of Proposition 2} \label{app:seperation}
In the proportional optimization problem (\ref{Latency_min}), the task assignment strategy $\textbf{s}$, computing capacity allocation $\bm{\theta}$ and transmission resources allocation $\bm{\phi}$ are coupled, expressed as

\begin{align}\label{Latency_initial}
&L(\bm{s},\bm{\theta},\bm{\phi})=\frac{\lambda_0^1}{\theta_0^1} + \sum_{n=1}^{N+1}\sum_{j=1}^{M_{n-1}}\sum_{i\in\mathcal{Q}_{n-1}^{j}}\left[ \frac{s_n^i \lambda_n^i}{\theta_n^i} + \frac{\rho s_n^i \lambda_n^i + (1-s_n^i)\lambda_n^i + \beta_n^i}{\phi_{n}^{j,i}} \right]
\end{align}

However, it is worth noting that the transmission resources of each node that allocated to its child nodes are limited, as described in (\ref{NBlockCon_AP4}), and the upperbound of the computing capacity of each node is fixed.
We consider that no spare computing capacity or transmission resource is left, i.e., (\ref{NBlockCon_ED2}) and (\ref{NBlockCon_AP2}) are satisfied.
Hence, by utilizing the Cauchy-Schwarz inequality \cite{CauchyInequality}, we can obtain the following inequation.
\begin{align}\label{CauchyInequa}
L(\bm{s},\bm{\theta},\bm{\phi})\geq L_{min}(\bm{s}) = &
 \left[\frac{\lambda_0^1}{\theta_{0}^{1,u}}\!+\!\sum_{n=1}^{N+1}\!\sum_{j=1}^{M_{n-1}}\!\!
 \sum_{i\in\mathcal{Q}_{n-1}^j}\!\!\frac{s_{n}^{i} \lambda_n^i}{\theta_{n}^{i,u}}\right] \!+ \nonumber\\
  & \sum_{n=1}^{N+1}\!\sum_{j=1}^{M_{n-1}}\!\frac{\left(\sum_{i\in\mathcal{Q}_{n-1}^{j}}
  \sqrt{ \rho s_n^i \lambda_n^i \!+\! (1\!-\!s_n^i)\lambda_n^i \!+\! \beta_n^i } \right)^2}{\phi_{n-1}^j},
\end{align}
where $\theta_{n}^{i,u}$ and $\phi_{n-1}^j$ are the boundary of the computing and transmitting capacity.
The proportional optimization problem (\ref{Latency_min}) is converted into a pure task assignment problem.

\vspace{15mm}
\section{Proof of the Proposition 4}\label{app:Concave}
We analyze the network with one parent node and $M$ child nodes. Let $s_i$ and $\lambda_i$ denotes the task assignment percentage and raw data arriving rate at child node $i$. The maximum computing capacity of child node $i$ is denoted by $\theta_i^u$, and the computing capacity of the parent node is denoted by $\theta^u$. The total transmission resource of the parent node is expressed by $\phi$.

The latency $L_{min}$ can be expressed by
$$L_{min}=\sum_{i=1}^{M}\frac{s_i \lambda_i}{\theta_i^u} + \frac{\sum_{i=1}^{M}(1-s_i)\lambda_i}{\theta^u} + \frac{\left(\sum_{i=1}^M\sqrt{(1-s_i)\lambda_i + \rho s_i \lambda_i}\right)^2}{\phi}.$$
The Hessian matrix of the system latency with $M$ child nodes can be expressed by
$$\textbf{H}_M =
\begin{bmatrix}
   h_{1,1} & h_{1,2} & \dots & h_{1,M} \\
   h_{2,1} & h_{2,2} & \dots & h_{2,M} \\
   \dots & \dots &   & \dots \\
   h_{M,1} & h_{M,2} & \dots & h_{M,M}
  \end{bmatrix}
  =
 \begin{bmatrix}
   \frac{\partial^2 L_{min}}{\partial s_1 s_1} & \frac{\partial^2 L_{min}}{\partial s_1 s_2} & \dots & \frac{\partial^2 L_{min}}{\partial s_1 s_M} \\
   \frac{\partial^2 L_{min}}{\partial s_2 s_1} & \frac{\partial^2 L_{min}}{\partial s_2 s_2} & \dots & \frac{\partial^2 L_{min}}{\partial s_2 s_M} \\
   \dots & \dots &   & \dots \\
   \frac{\partial^2 L_{min}}{\partial s_M s_1} & \frac{\partial^2 L_{min}}{\partial s_M s_2} & \dots & \frac{\partial^2 L_{min}}{\partial s_M s_M}
  \end{bmatrix}
$$
We can obtain the second partial derivative as follows
\begin{align}
  h_{i,i} &= \frac{\partial^2 L_{min}}{\partial s_i^2} = -Z\lambda_i^2 \frac{(\sum_{j=1}^M A_j) - A_i}{A_i^3}, \\
  h_{i,j} &= \frac{\partial^2 L_{min}}{\partial s_i s_j} = Z\lambda_i \lambda_j \frac{1}{A_i A_j},
\end{align}
where
\begin{align}
Z &= \frac{(1-\rho)^2}{2\phi} \geq 0, \quad \\
A_i &= \sqrt{(1-s_i)\lambda_i+\rho s_i \lambda_i} \geq 0.
\end{align}
Considering a normal vector $\textbf{x}=[x_1 x_2 \dots x_M]^T$, we obtain the following polynomial
\begin{equation}
  X_M(\textbf{x}) = \textbf{x}^T\textbf{H}_M\textbf{x} = \sum_{i=1}^{M}\sum_{j=1}^{M}h_{i,j}x_i x_j.
\end{equation}
We then prove that $X_M(\textbf{x})\leq0$ for any natural number $M$ by mathematical induction.
\begin{itemize}
\item When $M=1$: $X_1(\textbf{x})=0 \leq 0$
\item When $M=2$: $X_2(\textbf{x}) = -Z\frac{(A_1^2\lambda_2 - A_2^2\lambda_2)^2}{M_1^3 M_2^3}\leq 0$.
\item We assume that when $M = m-1$, $X_{m-1}(\textbf{x})= \textbf{x}^T\textbf{H}_{m-1}\textbf{x}\leq 0$.
\item Hence, when $M = m$, we have
\begin{align*}
  X_m(\textbf{x}) =& \textbf{x}^T\textbf{H}_{m}\textbf{x} =  \sum_{i=1}^{m}\sum_{j=1}^{m}h_{i,j}x_i x_j \\
  =& X_{m-1} - \frac{Z}{A_m^3}\sum_{i-1}^{m-1}\frac{(A_m^2\lambda_i-A_i^2\lambda_m)^2}{A_i^3} \leq 0.
\end{align*}
\end{itemize}
Since $X_M(\textbf{x})\leq0$ for any natural number $M$, the Hessian matrix $H_M$ is a seminegative definite matrix, implying that the function $L_{min}$ is concave\cite{Hessian}.

\par Moreover, the non-congested constraints are linear based on the \textbf{Proposition \ref{prop:non-congested}}. Hence, the the minimum value of a concave function is obtained at the vertex of the feasible set bounded by the non-congested constraints.
%After analyzing the sign of elements in the Hessian matrix of $L_{min}$, we note that the Hessian matrix is a seminegative definite matrix, implying that the function $L_{min}$ is concave \cite{Hessian}. Hence, the minimum value of a concave function is obtained at the vertex of the feasible set bounded by the non-congested constraints presented in \textbf{Proposition \ref{prop:non-congested}}.

\vspace{15mm}
\section{Proof of Remark 2}\label{app:complexity}
We consider the worst case, that the number of child nodes connected with each parent node is $Q=\max\limits_{0\leq n\leq N, 1\leq i\leq M_n}Q_n^i$. In this case, the number of nodes in the whole network can be calculated as below.
\begin{equation}
M = \sum_{n=0}^{N+1}Q^n=\frac{Q^{N+2}-1}{Q-1}.
\end{equation}
The complexity of the latency minimization algorithm is proportional to the number of feasible vertexes, which is proportional to the square of the number of the constraints and closely related to the complexity of finding the vertexes. The non-congested constraints derive from the computing capacity limitation of each node and the transmission resources limitation of each parent node.
The number of the computing capacity constraints equals that of all nodes:
\begin{equation}
K_c = M = \sum_{n=0}^{N+1}Q^n=\frac{Q^{N+2}-1}{Q-1}.
\end{equation}
The number of the transmission resource constraints equals the number of the parent nodes in the whole network, which can be expressed by
\begin{equation}
K_t = \sum_{n=0}^{N}Q^n=\frac{Q^{N+1}-1}{Q-1}.
\end{equation}
The number of the constraints is
\begin{equation}
K=K_c+K_t=\frac{Q^{N+2}+Q^{N+1}-2}{Q-1}.
\end{equation}
The maximum number of the vertexes of the feasible set is
\begin{align}
O(K^2)=\frac{O(Q^{2N+4})}{O(Q^2)}=O(M^2).
\end{align}

\end{appendices}

\end{document}